\documentclass[11pt]{article}
\usepackage[normalem]{ulem}
\usepackage{subcaption}
\usepackage{booktabs}
\usepackage{rotating}
\usepackage{array}
\usepackage{mathtools}
\usepackage{bm}
\usepackage{threeparttable}
\usepackage{enumitem}
\setlist{noitemsep}
\usepackage{amsthm,amssymb}
\usepackage{mathrsfs}
\usepackage{amsfonts}
\usepackage{amsmath}
\usepackage{graphicx}
\usepackage[longnamesfirst]{natbib}
\usepackage{amssymb}
\usepackage{comment}
\usepackage{multirow}
\usepackage[onehalfspacing]{setspace}
\usepackage{snaptodo}
\usepackage{float}
\usepackage[top=1.25in, bottom=1.25in, left=1.25in, right=1.25in]{geometry}
\usepackage{xcolor}
\usepackage{pdflscape}
\usepackage{bbm}
\usepackage{ragged2e}
\usepackage{floatpag}
\usepackage[hidelinks,colorlinks=true,linkcolor=black,citecolor=blue]{hyperref}
\usepackage{makecell}

\usepackage{snaptodo}
\snaptodoset{block rise=2em}
\snaptodoset{margin block/.style={font=\scriptsize}} 
\snaptodoset{chain bias=5cm} 
\newcommand{\ag}[1]{%
    \snaptodo[margin block/.append style=green!50!black]{%
	{\sloppy\textbf{Abhi}: #1}}}

 \newcommand{\mc}[1]{   \snaptodo[margin block/.append style=blue]{	\sloppy\textbf{Marghe}: #1}}
 
\usepackage{booktabs}
\newtheorem{theorem}{Theorem}

\newtheorem{definition}{Definition}

\newtheorem{lemma}{Lemma}

\numberwithin {equation}{section} 
\newtheorem{assumption}{Assumption}

\newcommand{\eigbig}{\overline{\textit{eig}}}
\newcommand{\eigsmall}{\underline{\textit{eig}}}

\theoremstyle{definition}

\numberwithin {equation}{section} 

\begin{document}


\begin{titlepage}

\title{A Nonparametric Test for Cross-Unit Spillovers\thanks{\textit{Acknowledgements}:  We are grateful to Marco Alfano, Aislinn Bohren, Arun Chandrasekhar, Ben Deaner, Guido Imbens, \'{A}ureo de Paula and Michael Vlassopoulos for helpful comments. We also thank audiences at UCL, Essex, Southampton, CFE London 2025, ISI Delhi ACEGD 2025, SEW/SEA Paris 2026 and the Celebrating James MacKinnon Conference Aarhus 2026. Comola acknowledges the financial support of the French National Research Agency (ANR-21-CE26-0002-01 and ANR-22-CE26-0001). Comunello and Gupta acknowledge the financial support of the Leverhulme Trust via grant RPG-2024-038. Gupta is supported in part by funding from the Social Sciences and Humanities Research Council of Canada.}}
\author{Margherita Comola\thanks{University Paris-Saclay (RITM) and Paris School of Economics. Email: margherita.comola@psemail.eu} \and Camila Comunello\thanks{Department of Economics, Universidad Torcuato Di Tella. Email:  camila.comunello@utdt.edu} \and Abhimanyu Gupta\thanks{Department of Economics, Queen's University, Dunning Hall, 94 University Avenue, Kingston K7L 3N6, Canada. Email:  abhimanyu.g@queensu.ca}}
\date{\today 
}

\maketitle
\begin{abstract}
\noindent 
Cross-unit dependence is pervasive in empirical applications and complicates econometric inference, especially when spillovers operate in nonlinear ways.
We propose a novel nonparametric test for cross-unit spillovers that may operate through peers’ attributes, peers’ outcomes, or both. The test is straightforward to implement, as it requires only estimation under the null hypothesis of no cross-unit spillovers, and is shown to have a convenient asymptotic standard normal distribution. It is also versatile, accommodating data generated by a wide range of interaction structures. We present four empirical illustrations showing that the proposed test can yield substantively different conclusions about the presence of cross-unit spillovers than existing approaches.
\\
\vspace{0in}
\vspace{0in}
\bigskip
\end{abstract}
\setcounter{page}{0}
\thispagestyle{empty}
\end{titlepage}
\pagebreak \newpage


\section{Introduction}

Cross-unit dependence is pervasive in economic and social interactions and poses a fundamental challenge for econometric inference.\footnote{\cite{VivianoRudder2024} report that approximately 40\% of experimental papers published in top-five economics journals in 2020 discuss spillovers as a potential threat to identification.}
This
 paper deals with the two most common sources of cross-unit spillovers, which arise through the attributes channel and the outcome channel. To illustrate, consider the canonical  cross-sectional linear model, where the outcome of the unit (e.g. individual) is regressed on own attributes. 
 First, the attributes of other units (`peers') may affect $i$'s outcome. We term this as \textbf{`covariate ($\boldsymbol{c}$) spillovers'}.  The second source of cross-unit dependence refers to the case where the outcomes of other units affect $i$'s outcome. We call this \textbf{`outcome ($\boldsymbol{y}$) spillovers'}. In this paper, we propose a test for unknown nonparametric spillovers operating through one or both channels, establish its asymptotic properties, and illustrate its applicability in four diverse settings. We term this the $\boldsymbol{s}$ \textbf{test}.

Cross-unit spillovers have received considerable attention from applied economists in a broad range of  contexts. These include, \emph{inter alia}, disease transmission \citep{MiguelKremer2004, Ozier2018}, educational outcomes \citep{ sacerdote2001peer, LaliveCattaneo2009, BoboniFinan2009,AvvisatiEtAl2013}, employment decisions \citep{DufloSaez2003, BrownLaschever2012} and technology adoption \citep{OsterThornton2012,BanerjeeEA2013,CaiEA2015}.  

A strand of the literature on (broadly defined) `peer effects'  has explicitly modeled cross-unit dependence in observational data via the attribute and/or the outcome channel, depending on the setting. Oftentimes, cross-unit dependence is modeled solely through the attribute channel, even though outcome spillovers could also be incorporated due to economic considerations.\footnote{The exclusion restriction that peers’ attributes serve as a reduced-form sufficient statistic for their outcomes is frequently imposed. However, when the research design permits, spillovers operating through both peers’ covariates and peers’ realized outcomes can be jointly identified, offering a sharper understanding of the underlying economic mechanisms \citep{BDF2009, BursztynFiorin2017}.}
  A related line of work focuses on treatment-mediated spillovers, which are a first-order concern in the context of impact evaluation as they violate the Stable Unit Treatment Value Assumption (SUTVA), which asserts that an individual’s potential outcomes should be independent of peers’ treatment assignments.
  In response to this concern, it has become increasingly common to design cluster-randomized experiments generating exogenous variation in peers’ treatment status.\footnote{Cluster-randomized trials (also known as `two-stage randomization experiments', `randomized-saturation experiments', `partial-population experiments') randomly assign different treatment rates across different clusters.}

A practical challenge is that spillovers are typically modelled through linear functions of peers' attributes and outcomes. While convenient, such specifications may fail to detect economically meaningful forms of cross-unit dependence when agents respond nonlinearly to their social or economic environment. In many settings, behavior depends on the distribution of peers' characteristics rather than on simple averages. For example, a small fraction of adopters may induce complementary behavior, whereas widespread adoption may generate substitutability.\footnote{Nonlinear adoption dynamics have been documented both theoretically and empirically \citep{bandiera2006social, young2009innovation, acemoglu2016networks}.} As a result, spillover effects may vary across the distribution and partially offset one another, causing linear specifications to understate or fail to detect cross-unit dependence.


We develop a nonparametric test for cross-unit dependence arising through the attribute channel, the outcome channel, or both. The test is inspired by classical Lagrange Multiplier (LM) diagnostics, such as the RESET test. We construct a test statistic for the null of no spillovers against flexible nonparametric alternatives, which are approximated using a series expansion. Because the test follows an LM approach, estimation is required only under the null. As a result, the procedure accommodates rich forms of cross-unit dependence while requiring only the estimation of a familiar multiple linear regression model.
We establish that the test is asymptotically standard normal under the null and is consistent. 
These results are derived under a cluster-robust framework, allowing for forms of within-cluster dependence commonly encountered in applied work. 
Extensions to alternative error dependence structures—such as serial correlation or more general forms of spatial dependence—are conceptually straightforward.


Our test is far-reaching in that it is versatile in its data requirements, which is an important advantage. First, it accommodates data defined through a variety of interaction structures, including those based on blocks or links. Block-type data are partitioned into separate self-exclusive groups within which all units are assumed to interact. 
Blocks may represent villages or schools or nuclear households for individuals,  geographical area and/or productive sectors for firms. 
Alternatively, network-type data contain detailed links between units, which may or may not overlap  (e.g. $i$ is linked to $j$ and $j$ is linked to $k$, but $i$ is not linked to $k$). This is the case for self-declared link data in household surveys, or trade data among firms from administrative records.
 Our test accommodates both data structures. Second, it allows for heterogeneous spillovers via multiple interaction matrices, as we justify in Appendix \ref{sec:app_ext_mult}. Third, it allows the interaction structure to be incomplete or measured noisily. In applied work, interaction data are often measured poorly yet still convey useful information. By embedding our test within a latent-space framework, Appendix \ref{sec:app_ext_netform} states high level conditions under which the perturbations induced by measurement error are asymptotically negligible, ensuring that inference remains valid.

Most of the methodological literature on cross-unit dependence focuses on the estimation of causal parameters under interference. By and large, existing approaches rely on parametric assumptions to achieve this aim \citep{Rosenbaum2007,HudgensHalloran2008, Hirano2010, LiuHudgens2014, BairdEtAl2018, Arduini2020, McNealisMoodieDean2024_JRSSC, viviano_etal2024}. However, more recently nonparametric and semiparametric methods are being increasingly adopted.  \cite{VazquezBare2023} works in a nonparametric identification framework to propose estimators for spillover effects in experiments, while \cite{viviano2025} studies policy targeting with interference in a semiparametric framework. A semiparametric approach is also used in \cite{Hoshin2023}, which develops treatment effects models with strategic interaction in treatment decisions, while \cite{Hoshino2024} use nonparametric methods to conduct causal inference under unknown interference structure and in more recent work provide a specification test for the latter \citep{Hoshino2026}. 

By proposing a nonparametric test for cross-unit spillovers, this paper complements the existing literature by providing practitioners with an easy-to-implement diagnostic tool to guide the design, validation, and refinement of estimation strategies. If no spillovers are detected, the test provides empirical support for ruling out interference concerns and proceeding with standard estimation strategies—such as linear intent-to-treat regressions. If instead the test detects nonlinear cross-unit dependence, the estimation strategy should be adapted to account for its potentially substantial effects.\footnote{While our test does not explicitly suggest the nonparametric functional form which best describes the data at hand, a variety of suitable econometric methods are available, see e.g. \cite{Jenish2012a, Jenish2016, Xu2015, Xu2018}.} Our test can also be useful prior to rolling out a large-scale survey, when researchers, based on the pilot, must decide whether and how to adjust the design to account for cross-unit dependence.\footnote{For instance, a researcher may apply our test to interaction data collected in a pilot study to assess whether detailed link information is required, or whether the treatment intensity should be exogenously varied across clusters in a subsequent scale-up. Because both design choices can entail substantial additional survey costs, they are warranted primarily when spillovers are expected to play an important role.} 

We illustrate the implementation and economic relevance of our test through four empirical applications that revisit spillovers in distinct settings. In the first illustration, we examine peer effects in a high-skill environment, namely professional golf tournaments \citep{GuryanKroftNotowidigdo}. 
The second studies the impact of interracial roommate assignment on stereotype attitudes and academic performance at the University of Cape Town \citep{CornoLaFerraraBurns}. The third examines a mixed-ability randomized deskmate intervention on student performance in Chinese elementary schools \citep{wuzhangwang2023}. The fourth revisits peer effects in academic research among French economists \citep{bosquetcombesetal2022}.
Across all four settings, we show that our test is able to detect cross-unit dependence through peers’ attributes and/or peers’ outcomes in many cases where linear functional forms fail to do so. To paraphrase a longstanding criticism of nonparametric specification tests, they can be akin to using a nuclear weapon against an ant, as few parametric specifications survive their scrutiny. Our applications, however, show that our test is sufficiently discerning to not reject the null when there is little evidence of spillovers.

Section \ref{sec:method} introduces the test statistic, while Section \ref{sec:asymptotics} characterizes its asymptotic behavior. In Section \ref{sec:implementation} we provide an implementation guide that includes advice on choosing the tuning parameters, and in Section \ref{illustrations} we illustrate the test by revisiting four existing studies. Section \ref{Conclusions} concludes. We show that our test controls size well in a Monte Carlo study of finite performance in Appendix \ref{app:sims}. In Appendix \ref{sec:app_ext} we extend our method to heterogeneous cross-unit dependence, and to embedded graphs with noisy measurement and parametric modeling of the underlying link formation structure. 
Appendices \ref{appendix:proofs} and \ref{appendix:lemmas} contain theorem proofs and auxiliary lemmas respectively, while Appendix \ref{app_C} reproduces the original results for our empirical demonstrations.


\section{Method and test statistic}\label{sec:method}
We write the general model with both $ \boldsymbol {c}$ and $\boldsymbol {y}$ spillovers in scalar notation as
\begin{equation}\label{model_1}
	y_i=  f(w_i' y) + x_i^\prime \beta +\sum_{j=1}^l g_j(w_i'c_j)+ \epsilon_i, \ \ \ i=1,....,n,
\end{equation}
where $f(\cdot):\mathbb{R}\rightarrow\mathbb{R}$ is an unknown function that captures `outcome ($\boldsymbol{y}$) spillovers', and $g_j(\cdot):\mathbb{R}\rightarrow\mathbb{R}$ are also unknown functions that capture `attribute ($\boldsymbol{c}$) spillovers'. $y=\left(y_1,\ldots,y_n\right)'$ is an observed outcome vector and  $W=\left(w_1,\ldots,w_n\right)'$ is a social weight matrix representing cross-unit interactions that is either fixed or exogenous conditional on observables, and has zero diagonal.\footnote{Since our approach is cast within an instrumental-variables framework, it naturally accommodates endogeneity of the term $w_i' y$ whenever valid instruments are available. In the context of network-type data, it is compatible with instrumental-variable strategies designed to address outcome spillovers that are endogenous due to simultaneity \citep{BDF2009} or due to assortative link formation on unobservables \citep{jochmans2023}.} $x_i'$ is the $i$-th row of an $n\times k$ regressor matrix $X$ that can have endogenous elements as long as instruments are available, $c_j$, $j=1,\ldots,l$, are some $n\times 1$ vectors of exogenous regressors and $\epsilon_i$ is an unobserved disturbance.\footnote{In some applications, it is useful to write the model as \[
	y_i=  f\left(w_i' y^*\right) + x_i^\prime \beta +\sum_{j=1}^l g_j(w_i'c_j)+ \epsilon_i, \ \ \ i=1,....,n,
\] where $y^*$ measures the same outcome as $y$ but need not have the same entries. For instance, in a peer effects study, $y$ could be an academic outcome for a given subsample but $y^*$ could be the outcome for a different subsample.} The data are observed as $n_g$ observations in each of $g=1,\ldots,G$ clusters, so that $n=\sum_{g=1}^G n_G$. We take $n_g$ as fixed with $\sup_{g=1,\ldots,G}n_g<\infty$ and therefore $n\sim G$, i.e. our sample size grows like the number of clusters.  Our approach gives rise to three different testing options:
\begin{enumerate}
\item Jointly test for both $\boldsymbol {c}$ and $\boldsymbol {y}$ spillovers (the `$\boldsymbol{cy}$ test').

\item Omit $f(w_i'y)$ from the general model (\ref{model_1}) and test for only $\boldsymbol {c}$ spillovers (the `$\boldsymbol{c}$ test').

\item Omit $g_j(s)=0,\;j=1,\ldots,l,$ from the general model (\ref{model_1}) and test for only $\boldsymbol {y}$ spillovers (the `$\boldsymbol{y}$ test').
\end{enumerate}
\noindent The corresponding null hypotheses are: 
\begin{eqnarray}
&&\boldsymbol{cy}\text{ test, }\mathcal{H}_0: f(s)=0 \text{ and } g_j(s)=0,\;j=1,\ldots,l,\label{nullxy}\\
&& \boldsymbol{c}\text{ test, }	\mathcal{H}_0: g_j(s)=0,\;j=1,\ldots,l, \label{nullx}\\
&& \boldsymbol{y}\text{ test, }		\mathcal{H}_0: f(s)=0,\label{nully}
\end{eqnarray}
for all $s\in support(s)$, which implies in all cases that the null model is $y_i=x_i'\beta+\epsilon_i$, i.e. a standard linear regression. Given these three possible testing choices, we combine them to define our recommended $\boldsymbol{s}$ test, which has the following rule-of-thumb decision rule (on which we elaborate in Section \ref{sec:implementation}):

\begin{definition}\label{def:s-test}(Rejection rule of the $\boldsymbol{s}$ test) Reject the null hypothesis of no spillovers if at least one of the $\boldsymbol{c}$, $\boldsymbol{y}$ or $\boldsymbol{cy}$ tests rejects the null hypothesis.
\end{definition}

The $\boldsymbol{cy}$ test jointly includes nonlinear spillovers in both the covariate/attribute and outcome channels, while the $\boldsymbol{c}$ test and $\boldsymbol{y}$ test examine the covariate and outcome channels individually, respectively. 
Most applied work focuses on linear spillovers in covariates and extending that to nonlinearity via the $\boldsymbol{c}$ test seems natural, while augmenting for outcome spillovers through the $\boldsymbol{cy}$ test demonstrates the full power of our approach.
For ease of exposition we present asymptotic results and notation for the most general case covered by the $\boldsymbol{cy}$ test in (\ref{nullxy}).

Let $\psi_i(s)$, for $i=1,\ldots,p$, be a user-chosen set of basis functions (our applications use Hermite polynomials) such that 
\begin{equation}\label{lambda_prime}
	f(s)=\sum_{i=1}^p \mu_{f,i} \psi_i(s)+r_f(s),\;\; g_j(s)=\sum_{i=1}^p \mu_{g_j,i} \psi_i(s)+r_{g_j}(s),\;j=1,\ldots,l,
\end{equation}
with $p=p_n$ a divergent deterministic sequence, i.e. $p\rightarrow\infty$ as $n\rightarrow\infty$, $\mu_f= (\mu_{f,1},\ldots, \mu_{f,p})^\prime$, $\mu_{g_j}= (\mu_{g_j,1},\ldots, \mu_{g_j,p})^\prime$  vectors of unknown series coefficients and $r_f(s)$, $r_{g_j}(s)$ approximation errors. We define our approximate null hypothesis as
\begin{equation}\label{approximate_null}
	\mathcal{H}_{0A}: \mu_f=0 \text{ and }\mu_{g_j}=0,\; j=1,\ldots,l, \text{ for some } \beta,
\end{equation}
which is a set of $q=p(l+1)$ restrictions, so that $q\rightarrow\infty$ as $n\rightarrow\infty$. This indicates that standard fixed-dimension asymptotics will not work for our test. Our test statistic is based on determining if the moment conditions for the instrumental variables (IV) estimate of $\beta$ under the null hypothesis are close enough to zero. OLS is of course a special case of this. 

Now, for each $i=1,\ldots,p$, define the $n\times 1$ vector $\Upsilon_{f,i}(y)= \left(\psi_i(w_1'y),\ldots,\psi_i(w_n'y)\right)'$
and the $n\times 1$ vectors $\Upsilon_{g_j,i}(c_j)= \left(\psi_i(w_1'c_j),\ldots,\psi_i(w_n'c_j)\right)'$, and write
\[
\Upsilon_{f,g}=\begin{pmatrix} \Upsilon_{f,1}(y)& \ldots& \Upsilon_{f,p}(y)&\Upsilon_{g_1,1}(c_1)& \ldots& \Upsilon_{g_1,p}(c_1)&\ldots& \Upsilon_{g_l,1}(c_l)& \ldots& \Upsilon_{g_l,p}(c_l) \end{pmatrix},
\]
which is an $n\times q$ matrix. Denote $\mu=(\mu_f',\mu_{g_1}',\ldots,\mu_{g_l}')'$. Given  the expansion in (\ref{lambda_prime}), the series approximated IV objective function is
\begin{equation}\label{2slsobj}
	\mathcal{F}_p( \beta, \mu, y)= \frac{1}{n}\left( y-\Upsilon_{f,g}\mu-X\beta\right)^\prime \mathcal{P}_Z\left( y-\Upsilon_{f,g}\mu-X\beta\right),
\end{equation}
where $Z$ is an $n\times m$ matrix of valid instruments, with $m\geq q+k$, and $\mathcal{P}_Z= Z(Z^\prime Z)^{-1} Z^\prime$. Next, define the $n\times (q+k) $ matrix  $U= \begin{pmatrix} \Upsilon_{f,g}& X \end{pmatrix}$. OLS is a special case with $Z=U$. 

Define the $(q+k)\times 1$ gradient vector $\tilde{d}( \beta, y)$ of (\ref{2slsobj}) under $\mathcal{H}_{0A}$ as
\begin{equation}\label{d}
	\tilde{d}(\beta, y)= \left.\frac{\partial \mathcal{F}( \mu, \beta, y)}{\partial (\mu, \beta)^\prime } \right\vert_{\mu=0}= -\frac{2}{n}U^\prime \mathcal{P}_Z  (y -  X\beta).
\end{equation}
Denoting by   $\hat{\beta}$ some consistent estimate of $\beta$, e.g. IV or OLS, under $\mathcal{H}_{0A}$,  the gradient evaluated at the corresponding residuals is
\begin{equation}\label{dhat}
	\hat{d}= \tilde{d}\left( \hat{\beta}, y\right)=  -\frac{2}{n}U^\prime \mathcal{P}_Z  (y - X\hat{\beta}).
\end{equation}
Let  
$\hat{J}=n^{-1}Z^\prime U$, where $\hat{J}$ is $m \times (q+k) $. Next, define the $m\times m$ matrices $\hat{M}= n^{-1}{Z^\prime Z}$ and $\hat{\Phi}=n^{-1}{Z^\prime \hat{\Sigma} Z}$, with $\hat{\Sigma}=diag\left(\hat{\Sigma}_1,\ldots,\hat{\Sigma}_G\right)$ where $\hat{\Sigma}_g$ has typical $(i,j)$-th element $\hat{\epsilon}_i\hat{\epsilon}_j$ and $\hat{\epsilon}_i= y_i- x_i^\prime\hat{\beta}$, for $i,j=1,\ldots,n_g$ and $g=1,\ldots,G$. We write
\begin{equation}\label{Hhat}
	\hat{H}=4 \hat{J}^\prime \hat{M}^{-1}\hat{\Phi}\hat{M}^{-1} \hat{J},
\end{equation}
and define
our cluster robust test statistic as
\begin{equation}\label{statistic}
	\mathcal{S}= \frac{n\hat{d}^\prime \hat{H}^{-1} \hat{d} - q}{\sqrt{2q}}.
\end{equation}
This is a weighted measure of the distance of the gradient from zero, centred and rescaled to account for $q\rightarrow\infty$.  In practice, especially for small $q$, one can use $n\hat{d}^\prime \hat{H}^{-1} \hat{d}$ as the test statistic with $\chi^2_q$ critical values instead of standard normal ones. We study this in our simulations in Appendix \ref{app:sims}.

	\section{Asymptotic theory}\label{sec:asymptotics}
We commence this section by introducing some technical assumptions to establish the limiting behaviour of (\ref{statistic}) under $\mathcal{H}_{0A}$. Throughout we denote by $K$ a generic positive constant, arbitrarily large but independent of $p$ and $n$.
\begin{assumption}\label{ass:errors} $\epsilon _{i}$ 
	are random
	variables with zero mean and unknown variance $\sigma_{i} ^{2}\in[c,K]$, $c>0$, and,
	for some $\tau >0,$ $ \mathbb{E} \left\vert\epsilon _{i}\right\vert^{8+\tau }\leq K$ for $i=1,\ldots,n$. Furthermore $\epsilon_{ig}$ and $\epsilon_{jg'}$ are independent for $g\neq g'$, $g,g'=1,\ldots,G$, while $\mathbb{E}\left(\epsilon_{ig}\epsilon_{jg}\right)=\sigma_{ijg}<\infty$, $i\neq j$, and we accordingly write ${\Sigma}=diag\left({\Sigma}_1,\ldots,{\Sigma}_G\right)$.
\end{assumption}

\begin{assumption}\label{ass:regressors}  $ \mathbb{E}(x_{ir}^4 )\leq K$ and $\mathbb{E}(z_{is}^4 )\leq K$, for $i=1,\ldots,n$ and $r=1,\ldots,k$ and $s=1,\ldots,l$.
\end{assumption}
\noindent We also allow $cov(\epsilon_i, x_{ij}) \neq 0$, for some $j=1,\ldots, k$, i.e. $X$ might contain some endogenous columns. Let $X_1$ be the $n \times k_1$ matrix containing the subset of exogenous columns of $X$, while $X_2$ ($n\times k_2$, with $k_2=k-k_1$) contains the endogenous ones. Now, for a generic symmetric positive-definite matrix $A$, let $\eigbig(A)$ and $\eigsmall(A)$ denote its largest and smallest eigenvalues, respectively. For a generic matrix $B$, denote by $\left\Vert B\right\Vert=\sqrt{\eigbig(B'B)}$, i.e. the spectral norm of $B$, and by  $\left\Vert B\right\Vert_\infty$ its largest absolute row sum.

%


\begin{assumption}\label{ass:eigsandinsts} 
The $n\times n$ matrix $\Sigma$ satisfies 
 \begin{equation}\label{ass_6_Sigma}
		\limsup_{n\rightarrow\infty}\sup_{g=1,\ldots,G}\eigbig(\Sigma_g) <\infty ,\;\;\; \liminf_{n\rightarrow\infty}\inf_{g=1,\ldots,G}\eigsmall\left(\Sigma_g \right) >0,
	\end{equation} 
the $m\times m$ matrix $M=\mathbb{E}(\hat{M})$, with $m\geq q+k$ , satisfies 
	\begin{equation}\label{ass_6_a}
		\limsup_{n\rightarrow\infty}\eigbig(M) <\infty ,\;\;\; \liminf_{n\rightarrow\infty}\eigsmall\left(M \right) >0,
	\end{equation}
	and the $(q+k) \times (q+k)$ matrix  
	$L=n^{-1}\mathbb{E}(U'U)$ satisfies
	\begin{equation}\label{ass_6_b}
		\limsup_{n\rightarrow\infty}\eigbig(L) <\infty ,\;\;\; \liminf_{n\rightarrow\infty}\eigsmall\left(L \right) >0,
	\end{equation}
	for $n$ large enough.  For some $\nu>0$ satisfying $n/p^{(\nu + 1/2)}=o(1)$, 
	\[
	\sup_z r_f(z)+\sup_{j=1,\ldots,l}\sup_z r_{g_j}(z)=O_p\left(p^{-\nu} \right),
	\] as $p\rightarrow\infty$.  $\mathbb{E}\left(u_{il_2}^4\right)\leq K$ for $i=1,\ldots,n, l_1=1,\ldots,m$ and $l_2=1,\ldots,q+k$, and $\epsilon_i$ and $z_j$ are uncorrelated for each $i,j=1,\ldots,n$.
\end{assumption}
\noindent Assumption \ref{ass:eigsandinsts} imposes regularity conditions and controls the approximation errors. Specifically, (\ref{ass_6_a})-(\ref{ass_6_b}) are asymptotic boundedness and no multicollinearity conditions for matrices of increasing dimension, while under Assumption \ref{ass:errors}, (\ref{ass_6_Sigma}) also ensures that $0<\sup_{i=1,\ldots,n}\Sigma_i< \infty$ in the special case of purely heteroskedasticity robust testing i.e. when $G=n$ and $\Sigma_i$ are scalars. For the instruments $Z$ we use at least $k_2$ columns of instruments for the endogenous covariates $X_2$, and also the columns of $X_1$, $WX_1$. We also use a set of instruments of the form $\psi_r\left(\sum_j w_{ij}x_{1, jl} \right)$, where $r=1,\ldots,p$, and $x_{1, jl}$ denotes the $(j,l)$th element of $X_1$, with $l=1,\ldots,k_1$.  For more discussion on approximation error decay rates see e.g. \cite{Chen2007}. Our next assumption sets a suitable bound on cross-sectional dependence, analogous to that in \cite{Lee2016}. Conditions such as linear process representations for the underlying random variables or the near-epoch dependence conditions of \cite{Jenish2012} imply that this assumption holds. 
 \allowdisplaybreaks
\begin{assumption}\label{ass:Mhat}
	Let
	\[
		\xi = \underset{0\leq l,k \leq m} {\sup} \ \underset{j\neq i}{\underset{i=1}{\overset{n}\sum} {{\underset{j=1}{\overset{n}\sum}}}}\left\vert cov (z_{il}z_{ik},z_{jl}z_{jk}) \right\vert,
		\varkappa =\underset{0\leq l \leq m,\;\; 0\leq r \leq q+k} {\sup} \ \underset{j\neq i}{\underset{i=1}{\overset{n}\sum} {{\underset{j=1}{\overset{n}\sum}}}}\left\vert cov(z_{il}u_{ir},z_{jl}u_{jr})\right\vert,
	\]
	and assume 
	\begin{equation}\label{delta_cond}
		\xi+\varkappa =O(n), \ \ \ \text{as} \ \ \ n\rightarrow \infty.
	\end{equation}
\end{assumption}
\noindent Our null asymptotic theory first approximates the test statistic $\mathcal{S}$ with a quadratic form in $\epsilon$, and then shows that this approximation is asymptotically standard normal. Write $J=E(\hat J)$ and define
\begin{align}
	d=d( \beta_0, y)=& -\frac{2}{n} J^\prime M^{-1/2}\left(I - M^{-1/2} N\left(N^\prime M^{-1} N \right)^{-1} N^\prime M^{-1/2} \right) M^{-1/2} Z^\prime \epsilon \notag \\=& -\frac{2}{n} J^\prime M^{-1/2}\mathcal{K}_{NM} M^{-1/2} Z^\prime \epsilon, 
\end{align}
where $\mathcal{K}_{NM}= \left(I - M^{-1/2} N\left(N^\prime M^{-1} N \right)^{-1} N^\prime M^{-1/2} \right)$ is $m\times m$ and $N= \mathbb{E}(\hat{N})$, with $\hat{N}= n^{-1} Z' X$, the last being  an $m \times k$ matrix with full rank under (\ref{ass_6_b}) in Assumption \ref{ass:eigsandinsts}.
Set
\begin{align}\label{H}
	H= n\mathbb{E}(d d^\prime)  = 4 J^\prime M^{-1/2} \mathcal{K}_{NM} M^{-1/2}\Phi M^{-1/2} \mathcal{K}_{NM} M^{-1/2} J,
\end{align}
with $\Phi =n^{-1}\mathbb{E}(Z^\prime \Sigma Z)$. Under Assumptions \ref{ass:errors} and \ref{ass:eigsandinsts}, $H^{-1}$ exists and is non-singular for $n$ large enough, using Lemma \ref{lemma:Omegaeigs}. We now state the main result of this section.
\begin{theorem}\label{theorem:nulldist}
	Let Assumption \ref{ass:errors} hold with the representation $\epsilon_{i}=\sum_{r=1}^n b_{ir}\eta_r$, where $\eta_r$ are i.i.d. mean zero and unit variance random variables and the $b_{ir}<K$ are finite constants that are non-zero only for the group that contains observation $i$. Also let $\mathcal{H}_0$ and Assumptions \ref{ass:regressors}-\ref{ass:Mhat} hold together with $\nu>5/2$ and $p^3/n =o(1)$. Then
	\begin{equation}
		\mathcal{S}\overset{d}{\rightarrow} N(0,1), \text{ as } n\rightarrow \infty.
	\end{equation}
\end{theorem}

\noindent Theorem \ref{theorem:nulldist} provides asymptotic justification for using one-sided, standard normal critical values as observed also by \cite{Hong1995}. The `linear process' representation of the errors is convenient to establish a general result for clustered error structure and has been used in more general settings \citep{Kelejian2007, Conley2023}. Our next theorem relates to the power properties of our test. First consider the global alternative
\begin{equation}
	\mathcal{H}_{1A}: \ \  \mu_i \neq 0, \ \ \ \text{for some} \ i=1,\ldots,q, \text{ and any } \beta,
\end{equation}
where $\mu_i$ denotes the $i$-th element of $\mu$. We introduce the unrestricted quantities
\begin{equation}\label{epsilonU_def}
	\epsilon_{Ui}(\mu, \beta)=  y_i - \mu^\prime \upsilon_{f,g,i}- \beta^\prime x_i,i=1,\ldots,n,  \ \ \ \text{and} \ \ \ \  \tilde{\Phi}_{U}= \tilde{\Phi}_{U}(\mu,  \beta)=n^{-1}{Z^\prime \tilde{\Sigma}_U Z},
\end{equation}
where $\upsilon_{f,g,i}'=\left(\psi_1(w_i'y),\ldots,\psi_p(w_i'y),\psi_1(w_i'c_1),\ldots,\psi_p(w_i'c_1),\ldots,\psi_1(w_i'c_l),\ldots,\psi_p(w_i'c_l)\right)'$,
$\tilde{\Sigma}_{U}=diag\left(\tilde{\Sigma}_{U1},\ldots,\tilde{\Sigma}_{UG}\right)$, and $\tilde{\Sigma}_{Ug}$ has $(i,j)$-th element $\epsilon_{Ui}(\mu, \beta)\epsilon_{Uj}(\mu, \beta)$. Then $\hat{\Phi}= \tilde{\Phi}_U(0_{q\times 1}, \hat{\beta})$. Let $\gamma = (\mu, \beta)\in\Gamma= \Re^q \times  \Re^k$ and introduce:
\begin{assumption}\label{ass:power} \textit{For all sufficiently large $n$ and all $j=1,\ldots,q+k$,}
	\begin{equation}\label{cond_omegaU_1}
		\underset{\gamma \in \Gamma} {\sup} \;\;\eigbig(\tilde{\Phi}_U)  +   \underset{\gamma \in \Gamma} {\sup}\;\; \eigbig\left(\frac{\partial\tilde{\Phi}_U}{\partial \gamma_j}\right) =O_p(1)  ,
	\end{equation}
	and
	\begin{equation}\label{cond_omegaU_2}
		\left\{ \underset{\gamma \in \Gamma} {\inf}\;\;\eigsmall(\tilde{\Phi}_U)\right\}^{-1}+ \left\{ \ \underset{\gamma \in \Gamma} {\inf}\;\;\eigsmall\left(\frac{\partial\tilde{\Phi}_U(\gamma)}{\partial\gamma_j}\right)\right\}^{-1}=O_p(1).
	\end{equation} 
\end{assumption}

%

\noindent This assumption imposes mild regularity under $\mathcal{H}_{1A}$, reminiscent of boundedness and invertibility conditions.  Our power result follows below.


\begin{theorem}\label{theorem:consistency}
	Under Assumptions \ref{ass:errors}-\ref{ass:power}, $\mathcal{H}_{1A}$, $\nu>5/2$, and $p^3/n =o(1)$, $\mathcal{S}$ provides a consistent test.
\end{theorem}

\section{Practical guidance for the $s$ test}\label{sec:implementation}

In this section, we discuss practical implementation issues and explain how we construct the recommended $\boldsymbol{s}$ test decision rule introduced in Definition \ref{def:s-test}.
\paragraph{Choice of tuning parameters.}
 Our testing procedure is nonparametric and as such requires the tuning parameter $p$ to be chosen. Given our rate condition $p^3/n\rightarrow 0$, a reasonable empirical choice might be $p=\max\left\{2,\left[n^{1/3}\right]\right\}$, where $[\cdot]$ denotes the closest integer, see e.g. \cite{Gupta2023}, and noting that we always want $p>1$ to distinguish our setting from the linear case when using polynomial bases. However, note that the rate $p^3/n\rightarrow 0$ is equivalent asymptotically to $p^3(l+1)/n\rightarrow 0$ because $l$ is fixed, but in finite samples the extra $l+1$ factor can play a role. Thus our recommendation is to use
\begin{eqnarray}
\text{For the }\boldsymbol{cy}\text{ test, with both } \boldsymbol {c} \text { and } \boldsymbol {y} \text{ spillovers}&:&  p_{cy}=\max\left\{2,\left[\frac{n^{1/3}}{l+1}\right]\right\}, \label{pxy_reco}\\	
\text{For the }\boldsymbol{c}\text{ test, with only } \boldsymbol {c} \text{ spillovers}&:& p_c=\max\left\{2,\left[\frac{n^{1/3}}{l}\right]\right\},\label{px_reco}\\
\text{For the }\boldsymbol{y}\text{ test, with only } \boldsymbol {y} \text{ spillovers}&:& p_y=\max\left\{2,\left[n^{1/3}\right]\right\}. \label{py_reco}
\end{eqnarray}
\noindent In Appendix \ref{sec:app_ext_mult} we also show that we can extend our test to a setting with multiple $W$ matrices, say $\ell$, for which we recommend $p_{cy,m}=\max\left\{2,\left[{n^{1/3}}/{\ell (l+1)}\right]\right\}$ in (\ref{pxym_reco}) therein for the $\boldsymbol{cy}$ test. 

\paragraph{Rank deficiency issues in $Z$.} In finite samples, the instrument matrix $Z$ may be rank deficient. In such cases, we recommend identifying a maximal set of linearly independent columns of $Z$, yielding an $n\times r$ submatrix with full column rank, where $r=rank(Z)$. The test statistics can then be constructed using this reduced instrument matrix. Standard numerical routines are available for extracting a linearly independent subset of columns, with restriction of course that $m\geq q+k$. 

\paragraph{Critical values.} While our asymptotic theory justifies the use of one-sided standard normal critical values for the $\boldsymbol{c}$, $\boldsymbol{y}$ and $\boldsymbol{cy}$ tests, it is reasonable to suspect that these might not be completely reliable for small $(n,p)$. A recommended robustness check is to also compare the test statistic $\mathcal{S}$ to $(\chi^2_{q,\alpha}-q)/\sqrt{2q}$, where $\chi^2_{q,\alpha}$ is the critical value at the $\alpha\%$ significance level for a $\chi^2_q$ distribution.

\paragraph{Discussion of the rejection rule in Definition \ref{def:s-test}.} Our recommended $\boldsymbol{s}$-test rejection rule in Definition \ref{def:s-test} warrants further discussion. Our methodology yields three candidate tests for spillovers: through individual covariates ($\boldsymbol{c}$), through outcomes ($\boldsymbol{y}$), and through both channels jointly ($\boldsymbol{cy}$). One may therefore ask why we do not simply recommend one of these tests, rather than the composite decision rule underlying the $\boldsymbol{s}$ test in Definition \ref{def:s-test}. The reason is that reliance on any single channel can be misleading. To see this, consider the following cases:
\begin{enumerate}
    \item \textit{At least one of the $\boldsymbol{c}$ or $\boldsymbol{y}$ tests reject the null of no spillovers, but the $\boldsymbol{cy}$ test does not.} 
    This case can arise when one out of the covariate or outcome channels exhibits strong spillovers but the other does not, or does so to a much weaker degree. If the strength of the channel through which spillovers occur is sufficiently weak relative to the no-spillover channel, the combined $\boldsymbol{cy}$ test statistic can become too small and mask the spillover. Degrees of freedom issues can also play a role; indeed, the $\boldsymbol{cy}$ test has $q=p(l+1)$ while the other tests have $q=p$ or $q=pl$. This can distort the test procedure in finite samples. Our $\boldsymbol{s}$ test rejection rule avoids these problems because it will detect spillovers via the $\boldsymbol{c}$ or $\boldsymbol{y}$ tests. \vspace{0.2cm}
    \item \textit{Neither the $\boldsymbol{c}$ or $\boldsymbol{y}$ tests reject the null of no spillovers, but the $\boldsymbol{cy}$ test does.} This can be caused by a high degree of collinearity between the series approximations of $f(s)$ and the $g_j(s)$, namely the columns of the matrix $\Upsilon_{f,g}$. Our theory rules out perfect multicollinearity, nevertheless a high degree of collinearity can make the $\boldsymbol{c}$ and $\boldsymbol{y}$ test statistics small and therefore unreliable on their own as a decision tool. Our $\boldsymbol{s}$ test rejection rule avoids this problem of high collinearity masking the spillover because it will detect the spillover via the $\boldsymbol{cy}$ test.
\end{enumerate}

\allowdisplaybreaks
\section{Four demonstrations}\label{illustrations}

In this section, we illustrate the scope of our nonparametric test by revisiting four existing studies that investigate peer effects across different contexts—high-skill professionals, interracial college roommates, elementary-school deskmates, and academic researchers—with mixed findings. We use Hermite polynomials as basis functions in all four examples. Our test detects spillovers in many cases where linear specifications fail to do so, thereby overturning several, but not all, of the original findings of no spillovers. These applications highlight the repercussions of ignoring nonlinearities in the spillover mechanism. Additional details on the original studies are reported in Appendix \ref{app_C}. 



\subsection{Professional golf tournaments \citep{GuryanKroftNotowidigdo}} \label{golf}\label{applications}


Our first example builds on \cite{GuryanKroftNotowidigdo}, who study whether peer effects influence individual productivity in high-skill professional environments in the context of professional golf tournaments. They exploit a natural experiment within golf tournaments where playing partners are randomly assigned within predefined block–round categories. 

Their original results are reproduced in Appendix \ref{app_C}, Table \ref{tab:GKNaugmented}, and include three specifications.\footnote{The large sample size and relatively few covariates imply large values for $p$ using (\ref{pxy_reco})-(\ref{py_reco}). This can cause rank deficiency issues in $Z$, practical guidance for which was provided in Section \ref{sec:implementation}.}  In specification (i), players' performance is modelled as a function of their own ability, measured by the corrected handicap score, and the average ability of their peers in the same block–round–tournament.\footnote{
For details about the way the corrected handicap score is calculated, see Appendix \ref{app_C}.} Specification (ii)  incorporates alternative measures of peer ability—average driving distance, number of putts, and number of greens hit—designed to distinguish motivation effects (e.g., higher effort induced by stronger partners) from learning effects (e.g. adapting to observed putting strategies). Specification (iii) introduces heterogeneity by interacting partners’ average ability with a player’s own baseline ability and years of professional experience. 

While previous studies have found significant positive peer effects in low-skill labour markets \citep{BandieraBarankayRasul,MasMoretti}, \cite{GuryanKroftNotowidigdo} find limited evidence of peer effects in individual performance. In particular, they conclude against peer effects in specifications (i) and (ii), while specification (iii) provides some support for heterogeneous peer effects via the experience channel.


\subsubsection*{Test results}


Table \ref{tab:GKNjointtest} summarizes the authors' main results and the results from our tests.
The three rows correspond, in order of appearance, to the three original specifications reported in Appendix \ref{app_C}, Table \ref{tab:GKNaugmented}.
Columns 1--2 report the specification number and the construction of the attribute peer exposure variables $w'c$, respectively.
The `Original result' column reports the conclusions reached in the original paper regarding the null hypothesis of no linear peer effects, while the `$\boldsymbol{s}$ test result' column reports the conclusion of our test.
The `Rejection channel' column indicates the channels through which we reject the null: peers' attributes ($\boldsymbol{c}$), peers' outcomes ($\boldsymbol{y}$), or both ($\boldsymbol{cy}$).
Columns `$n$' and `$l$' report the sample size and the number of peer attribute terms used in the test, respectively.
Finally, the column `$p_c$, $p_y$, and $p_{cy}$' reports the number of basis functions used to test each spillover channel.

\begin{table}[p]\centering
\caption{Peer Exposure Design and Test Results: \citet{GuryanKroftNotowidigdo}}
\label{tab:GKNjointtest}
\small
\renewcommand{\arraystretch}{1.25}
\resizebox{\textwidth}{!}{%
\begin{tabular}{
c
p{3cm}
c
c
c
c
c
c
}
\toprule
\multicolumn{8}{l}{\textbf{Dependent variable:} Score in a given tournament-round, $y_{i,tr}$ } \\
\multicolumn{8}{l}{\textbf{Peer exposure in outcomes:} $w_{i,tr}'y_{tr}$} \\
\addlinespace
\midrule

Spec. & Peer exposure $w'\boldsymbol{c}$
& Original result
& $\boldsymbol{s}$ test result
& Rejection channel
& $n$
& $l$
& $p_c$, $p_{y}$, $p_{cy}$ \\
\midrule

(i)
& $w_{i,tr}'Ability$
& Do not reject
& Reject
& $\boldsymbol{c}$, $\boldsymbol{y}$, $\boldsymbol{cy}$
& 17,492
& 1
& 26, 26, 13\\ [2 em]

(ii)
& $\begin{aligned}[t]
& w_{i,tr}'DrivDist, \\
& w_{i,tr}'Greens, \\
& w_{i,tr}'Putts \\
& \end{aligned}$
& Do not reject
& Reject
& $\boldsymbol{c}$, $\boldsymbol{y}$, $\boldsymbol{cy}$
& 17,182
& 3
& 9, 26, 6 \\ [5em]

(iii)
& $\begin{aligned}[t]
& w_{i,tr}'Ability, \\
&  Ability_i \times w_{i,tr}' Ability,\\
& Exp_i \times w_{i,tr}' Ability\\
& \end{aligned}$
& Reject
& Reject
& $\boldsymbol{c}$, $\boldsymbol{y}$, $\boldsymbol{cy}$
& 17,492
& 3
& 9, 9, 4 \\
\bottomrule
\end{tabular}
}

\footnotesize
\justifying
\textit{Note:} 
Each row corresponds to a specification in Table \ref{tab:GKNaugmented} in Appendix \ref{app:gkn}. 
Columns 1-2 report the specification number and the construction of attribute peer exposure variables. The `Original result' column reports the conclusions reached in the original paper for the corresponding null hypothesis of no linear peer effects. The $\boldsymbol{s}$ test examines dependence operating through three rejection channels: the $\boldsymbol{c}$ channel through peers’ attributes, the $\boldsymbol{y}$ channel through peers’ outcomes, and the $\boldsymbol{cy}$ channel through both peers’ attributes and outcomes. Column $n$ reports the sample size and column $l$ the number of peer attribute terms. The column `$p_c$, $p_y$, $p_{cy}$' reports the selected number of basis functions used in each test. The outcome variable is the golf score for the round. $Ability_i$ is measured by the player's average handicap, $DrivDist_i$ is the average driving distance, $Greens_i$ is the average number of greens hit in regulation, and $Putts_i$ is the average number of putts per round, all averaged over the previous 2-3 years. $Exp_i$ is measured as years of experience. Peer exposure is measured using a social vector $w_{i,tr}$, where each player is exposed to the weighted average of the attributes/outcomes of peers within the same group-tournament-round. Sample weights are given by the inverse of the sample variance of the estimated ability of each player, in line with the original study. All specification controls are identical to those reported in Table \ref{tab:GKNaugmented} in Appendix \ref{app:gkn}. Tests are performed at the 95\% confidence level. Standard errors are clustered at the playing group level.
\end{table}

Our findings suggest the presence of peer effects in professional golf that may not be fully captured by the linear specifications in \cite{GuryanKroftNotowidigdo}. Indeed, our $\boldsymbol{s}$ tests detect spillovers in all three specifications, pointing to nonlinear effects that the original specifications fail to detect in cases (i) and (ii).

\subsection{Interracial contact, stereotypes, and academic performance \citep{CornoLaFerraraBurns}}\label{roommates}

In our second demonstration we revisit \cite{CornoLaFerraraBurns}, who exploit a policy implemented by the University of Cape Town to study whether interracial interaction affects stereotypes, attitudes, and academic performance in post-apartheid South Africa. The policy randomly allocates first-year students to roommates, providing exogenous variation in whether a student shares a double room with a roommate of a different race.

We focus on the analysis of whether mixed-race interaction affects academic performance.  The authors estimate the effect of a mixed room intervention on four academic outcomes --- GPA, the number of exams passed, eligibility to continue to the second year, and a composite index of academic performance --- separately for the White subsample, the Black subsample, and the full sample. Their original results are reproduced in Appendix \ref{app_C}, Table \ref{tab:CLBregress}. They find that Black students assigned to mixed-race rooms experience significant improvements in academic performance, whereas the estimated effects for White students are close to zero and statistically insignificant. In the full sample, the effects are positive and significant for all outcomes except GPA.


\subsubsection*{Test results}

Table \ref{tab:CLBtest} reports our $\boldsymbol{s}$ test results. Rows are ordered by panels A (Whites), B (Blacks) and C (full sample), and, within each panel, by specification (i)--(iv), matching the layout of Table \ref{tab:CLBregress} in Appendix \ref{app:clb}. Our nonparametric test reveals additional spillover patterns that complement the authors' original findings. 
The $\boldsymbol{s}$ test agrees with the authors in nine of the twelve cases, including two cases in Panel A where it preserves the original conclusion of `Do not reject'.  On the other hand, in three cases, the original linear regressions detect no significant peer effects, whereas our $\boldsymbol{s}$ test does.

\begin{table}[p]\centering
\caption{Peer Exposure Design and Test Results: \citet{CornoLaFerraraBurns}}
\label{tab:CLBtest}
\small
\renewcommand{\arraystretch}{1.25}
\resizebox{\textwidth}{!}{%
\begin{tabular}{
p{.9cm}
p{3.4cm}
c
c
c
c
c
c
}
\toprule
\multicolumn{8}{l}{\textbf{Peer exposure in attributes:} $w_{i}'Race$} \\
\multicolumn{8}{l}{\textbf{Peer exposure in outcomes:} $w_{i,\text{race}}'y$} \\
\multicolumn{8}{l}{Panels: A = Whites, B = Blacks, C = Full sample} \\
\multicolumn{8}{l}{$l=1$} \\
\addlinespace
\midrule

Spec.
& $y$
& Panel
& Original result
& $\boldsymbol{s}$ test result
& Rejection channel
& $n$
& $p_c$, $p_{y}$, $p_{cy}$
\\
\midrule

(i)
& GPA
& A
& Do not reject
& Do not reject
& --
& 117
& 5, 5, 2
 \\ [.5em]

(ii)
& No. of exams passed
& A
& Do not reject
& Reject
& $\boldsymbol{y}$
& 117
& 5, 5, 2 \\ [.5em]

(iii)
& Eligible to continue
& A
& Do not reject
& Reject
& $\boldsymbol{c}$
& 117
& 5, 5, 2 \\ [.5em]

(iv)
& Academic perf. index
& A
& Do not reject
& Do not reject
& --
& 117
& 5, 5, 2 \\
\addlinespace

(i)
& GPA
& B
& Reject
& Reject
& $\boldsymbol{c}$, $\boldsymbol{cy}$
& 332
& 7, 7, 3 \\ [.5em]

(ii)
& No. of exams passed
& B
& Reject
& Reject
& $\boldsymbol{c}$, $\boldsymbol{y}$, $\boldsymbol{cy}$
& 332
& 7, 7, 3
 \\ [.5em]

(iii)
& Eligible to continue
& B
& Reject
& Reject
& $\boldsymbol{c}$, $\boldsymbol{cy}$
& 332
& 7, 7, 3\\ [.5em]

(iv)
& Academic perf. index
& B
& Reject
& Reject
& $\boldsymbol{c}$, $\boldsymbol{cy}$
& 332
&  7, 7, 3 \\
\addlinespace

(i)
& GPA
& C
& Do not reject
& Reject
& $\boldsymbol{cy}$
& 499
& 8, 8, 4
 \\ [.5em]

(ii)
& No. of exams passed
& C
& Reject
& Reject
& $\boldsymbol{y}$
& 499
& 8, 8, 4 \\ [.5em]

(iii)
& Eligible to continue
& C
& Reject
& Reject
& $\boldsymbol{y}$
& 498
& 8, 8, 4 \\ [.5em]

(iv)
& Academic perf. index
& C
& Reject
& Reject
& $\boldsymbol{y}$, $\boldsymbol{cy}$
& 498
& 8, 8, 4\\

\bottomrule
\end{tabular}
}

\justifying
\footnotesize

\textit{Note:} Each row corresponds to one panel $\times$ specification cell of Table \ref{tab:CLBregress} in Appendix \ref{app:clb}. Panels A--C are the White, Black, and full samples, respectively; specifications (i)--(iv) use GPA, number of exams passed, eligibility to continue, and the academic performance index as the dependent variable. The `Original result' column reports the conclusions reached in the original paper for the corresponding null hypothesis of no linear peer effects. The $\boldsymbol{s}$ test examines dependence operating through three rejection channels: the $\boldsymbol{c}$ channel through peers’ attributes, the $\boldsymbol{y}$ channel through peers’ outcomes, and the $\boldsymbol{cy}$ channel through both peers’ attributes and outcomes. Column $n$ reports the sample size.
The column `$p_c$, $p_y$, $p_{cy}$' reports the selected number of basis functions used in each test. The attribute variable $Race$ is the race indicator of each student, so that $w_i'Race$ is the race of $i$'s baseline roommate (see Footnote \ref{footnote:w_corno}).
All specification controls are identical to those reported in Table \ref{tab:CLBregress}. Tests are performed at the 95\% confidence level. Standard errors are clustered at the room level. 

\end{table}

\subsection{Ability mix and student performance \citep{wuzhangwang2023}}\label{wzw}

In our third demonstration we revisit \cite{wuzhangwang2023}, who study a randomized deskmate intervention in elementary schools in China. Exploiting the fixed-seat system, under which students sit beside the same deskmate throughout the semester, the authors randomly pair previously high- and low-achieving students and examine effects on academic performance and `big five' personality traits.

The design features two treatment arms. In mixed-seating (MS) classes, students above and below the median of a prior examination are randomly paired as deskmates for 20 weeks. In mixed-seating with reward (MSR) classes, the same pairing is combined with a tournament-type incentive for high-achieving students: they receive a monetary award and a certificate of commendation if their deskmate's score improvement ranks in the top 10\% among lower-track students in the class. Control classes receive random seat assignment. The authors find that MS alone does not raise test scores relative to control, whereas MSR raises low-achieving students' mathematics scores by about 0.24 standard deviations; high-achieving students are largely unaffected academically. The MSR intervention also increases extraversion and agreeableness for both tracks.

We focus on their deskmate-level peer-effect specifications, which regress endline outcomes on the deskmate's baseline achievement. These results are reproduced in Appendix \ref{app:wzw}, Table \ref{tab:peer_effects_mixed_seating}. Statistically significant linear effects of deskmate baseline performance are largely absent in MS classes and remain limited under MSR---appearing mainly for selected personality traits among lower-track students---with little evidence that deskmate baseline scores raise endline academic performance measured by z-scores.

\subsubsection*{Test results}

In Table \ref{tab:WZWtest} we apply our $\boldsymbol{s}$ test to the deskmate peer structure. Following the authors' setup we define peer attribute exposure as the deskmate's baseline measurement of the outcome of interest, $w_i'\text{y}^{base}$, and peer outcome exposure as the deskmate's endline outcome, $w_i'y^{end}$. Rows are ordered by panels A (MS/low), B (MS/high), C (MSR/low) and D (MSR/high) and, within each panel, by specification (i)--(vi), matching Table \ref{tab:peer_effects_mixed_seating} in Appendix \ref{app:wzw}. The remaining columns follow previous conventions. Our $\boldsymbol{s}$ test rejects the null in most cases where the original linear specifications do not. In Panel~A, the $\boldsymbol{s}$ test rejects for every specification, whereas the original regressions do not reject at the 5\% level. In Panel~B, we fail to reject only for academic z-score and Agreeableness, matching the authors on those two cases but rejecting for the remaining traits. In Panel~C, we do not reject for academic z-score and Neuroticism but reject for Agreeableness, as in the original study; we additionally reject for Extraversion, Openness, and Conscientiousness. In Panel~D, we reject throughout, while the original regressions do not reject at 5\%. Overall, the results capture nonlinear deskmate dependence for both reward and no reward settings.

\begin{table}[p]\centering
\caption{Peer Exposure Design and Test Results: \citet{wuzhangwang2023}}
\label{tab:WZWtest}
\footnotesize
\renewcommand{\arraystretch}{1.2}
\resizebox{\textwidth}{!}{%
\begin{tabular}{
c c c p{2.2cm} c c c
c     
c      
c             
p{2.2cm}      
c             
c             
}
\toprule
\multicolumn{7}{l}{\textbf{Peer exposure in attributes:} $w_{i}'y^{base}$ }\\
\multicolumn{7}{l}{\textbf{Peer exposure in outcomes:} $w_{i}'y^{end}$} \\
\multicolumn{7}{l}{Panels: A = MS/low, B = MS/high, C = MSR/low, D = MSR/high} \\
\multicolumn{7}{l}{$l=1$,  \hspace{0.5cm} $p_c$, $p_y$, $p_{cy} =$ 7, 7, 3} \\
\addlinespace
\midrule

Spec.
& $y$ 
& Panel 
& Original result
& $\boldsymbol{s}$ test result 
& Rejection channel
& $n$ 
 \\
\midrule

(i) 
& Z-score
& A 
& Do not reject 
& Reject  
& $\boldsymbol{c}$, $\boldsymbol{cy}$
& 317 
\\

(ii) 
& Extraversion  
& A 
&  Do not reject 
& Reject 
& $\boldsymbol{c}$
& 317 
 \\

(iii) 
& Agreeableness  
& A
&  Do not reject
& Reject 
& $\boldsymbol{cy}$
& 317 
  \\

(iv) 
& Openness 
& A 
& Do not reject 
& Reject 
& $\boldsymbol{c}$
& 317 
  \\

(v) 
& Neuroticism 
& A 
&  Do not reject 
& Reject 
& $\boldsymbol{c}$, $\boldsymbol{y}$
& 317 
  \\

(vi) 
& Conscient. 
& A 
&  Do not reject 
& Reject 
&  $\boldsymbol{c}$
& 317 
  \\
\addlinespace

(i) 
& $\text{z-score}_{end}$ 
& B 
& Do not reject 
& Do not reject 
& --
& 317 
 \\

(ii) 
& Extraversion  
& B 
& Do not reject 
& Reject 
& $\boldsymbol{c}$
& 317 
  \\

(iii)  
& Agreeableness  
& B 
&  Do not reject 
&  Do not reject 
& --
& 317 
  \\

(iv) 
& Openness 
& B 
&  Do not reject 
& Reject 
& $\boldsymbol{c}$, $\boldsymbol{y}$
& 317 
  \\

(v) 
& Neuroticism 
& B 
&  Do not reject 
& Reject 
& $\boldsymbol{c}$, $\boldsymbol{y}$, $\boldsymbol{cy}$
& 317 
  \\

(vi) 
& Conscient. 
& B
&  Do not reject 
& Reject 
&  $\boldsymbol{c}$
& 317 
  \\
\addlinespace

(i) 
& $\text{z-score}_{end}$ 
& C 
& Do not reject 
& Do not reject 
& --
& 297 
 \\

(ii) 
& Extraversion 
& C 
& Do not reject 
& Reject
& $\boldsymbol{c}$, $\boldsymbol{y}$, $\boldsymbol{cy}$
& 297 
  \\

(iii) 
& Agreeableness  
& C 
& Reject   
& Reject 
& $\boldsymbol{c}$, $\boldsymbol{cy}$
& 297      
  \\

(iv) 
& Openness 
& C 
& Do not reject 
& Reject
& $\boldsymbol{c}$, $\boldsymbol{cy}$
& 297 
  \\

(v) 
& Neuroticism 
& C 
&  Do not reject
&  Do not reject
& --
& 297 
  \\

(vi) 
& Conscient. 
& C 
&  Do not reject 
& Reject
&  $\boldsymbol{y}$, $\boldsymbol{cy}$
& 297 
  \\
\addlinespace

(i) 
& $\text{z-score}_{end}$ 
& D 
& Do not reject 
& Reject      
& $\boldsymbol{c}$
& 297 
  \\

(ii) 
& Extraversion  
& D 
&  Do not reject 
&  Reject 
& $\boldsymbol{y}$, $\boldsymbol{cy}$
& 297 
  \\

(iii) 
& Agreeableness  
& D 
& Do not reject 
& Reject 
& $\boldsymbol{cy}$
& 297 
  \\

(iv) 
& Openness 
& D 
&  Do not reject 
& Reject 
& $\boldsymbol{c}$, $\boldsymbol{y}$
& 297 
  \\

(v) 
& Neuroticism 
& D 
&  Do not reject 
& Reject 
&  $\boldsymbol{y}$
& 297 
  \\

(vi) 
& Conscient. 
& D 
&  Do not reject
& Reject 
&  $\boldsymbol{y}$
& 297 
  \\

\bottomrule
\end{tabular}
}

\justifying
\footnotesize
\textit{Note:} Each row corresponds to one panel $\times$ specification cell of Table \ref{tab:peer_effects_mixed_seating} in Appendix \ref{app:wzw}. Panels A--D are composed of lower- or upper-track students in the MS and MSR classes; specifications (i)--(vi) use academic z-score and the five personality traits at endline as the outcome of interest respectively. The `Original result' column reports the conclusions reached in the original paper for the corresponding null hypothesis of no linear peer effects. The $\boldsymbol{s}$ test examines dependence operating through three rejection channels: the $\boldsymbol{c}$ channel through peers’ attributes, the $\boldsymbol{y}$ channel through peers’ outcomes, and the $\boldsymbol{cy}$ channel through both peers’ attributes and outcomes. Column $n$ reports the sample size.
Tests are performed at the 95\% confidence level. Standard errors are clustered at the class level.
\end{table}

\subsection{Peer effects in academic research \citep{bosquetcombesetal2022}}\label{bch}

Our fourth demonstration focuses on \cite{bosquetcombesetal2022}, who study peer effects in academic research among economists in French universities. They exploit a national contest (\textit{concours d'agr\'egation}) that centrally allocates successful candidates across universities, generating quasi-exogenous variation in colleagues' field-specific productivity. Identification is sharpened by focusing on arrivals among the lowest-ranked successful candidates, whose remaining choice sets are tightly constrained, so that the field specialization of the arriving professor can be treated as approximately as good as random for the receiving university.

The authors find no peer effects when peers are defined as the entire university. Restricting attention to colleagues in the same JEL field 
 and year, however, they find that one additional publication by field peers raises own productivity by a substantial range of about  0.6 publications on average. Accounting for heterogeneity in peer interactions, they further show that spillovers are weaker for women and older researchers. Their main results are reproduced in Appendix \ref{app_C}, Table \ref{tab:peer_effects_ols}.

\subsubsection*{Test results}

Table \ref{tab:BCHtest} reports our nonparametric test for the specifications in Table \ref{tab:peer_effects_ols}. 
The attribute terms in $w'\boldsymbol{c}$ are the number of peers and the average peer outcome in a given university and year (restricted to the same JEL code from specification~(iii) onwards). The richer specifications~(iv)--(vi) include their interactions  with gender and/or age. 

Our $\boldsymbol{s}$ test rejects the null in every specification. This aligns with the authors for the specifications with field-level peer effects (iii)--(vi),  but differs for  specifications (i)--(ii), where the original regressions do not reject. 

\begin{table}[p]\centering
\caption{Peer Exposure Design and Test Results: \cite{bosquetcombesetal2022} }
\label{tab:BCHtest}
\small
\renewcommand{\arraystretch}{1.25}
\resizebox{\textwidth}{!}{%
\begin{tabular}{
c      
p{5.2cm}      
c             
c             
p{1.6cm}      
c             
c             
c             
}
\toprule
\multicolumn{8}{l}{\textbf{Dependent variable:} Individual research output } \\

\addlinespace
\midrule
Spec. 
& Peer exposure $w'\boldsymbol{c}$ 
& Original result
& $\boldsymbol{s}$ test result
& Rejection channel
& $n$ 
& $l$ 
& $p_c$, $p_y$, $p_{cy}$ 
\\

\midrule
(i)  
& $\begin{aligned}[t]
& w_{i,ut}'Peer, \\
& w_{i,ut}'Output
\end{aligned}$
& Do not reject
& Reject 
& $\boldsymbol{c}$, $\boldsymbol{y}$
& 42{,}861 
& 2 
& 17, 17, 11 \\ [3em]

(ii)   
& $\begin{aligned}[t]
& w_{i,ut}'Peer, \\
& w_{i,ut}'Output
\end{aligned}$
& Do not reject
& Reject 
& $\boldsymbol{c}$, $\boldsymbol{y}$,  $\boldsymbol{cy}$
& 771{,}498 
& 2 
& 46, 46, 31 \\ [3em]

(iii)   
& $\begin{aligned}[t]
& w_{i,ut}'Peer, \\
& w_{i,uft}'Output
\end{aligned}$
& Reject
& Reject
& $\boldsymbol{cy}$
& 771{,}498 
& 2 
& 31, 31, 23 \\ [3em]

(iv)  
& $\begin{aligned}[t]
& w_{i,ut}'Peer, \\
& w_{i,uft}'Output \\
& Woman_i \times w_{i,uft}'Output
\end{aligned}$
& Reject
& Reject
& $\boldsymbol{c}$, $\boldsymbol{y}$
& 771{,}498 
& 3 
& 31, 31, 23 \\ [4.5em]

(v) 
& $\begin{aligned}[t]
& w_{i,ut}'Peer, \\
& w_{i,uft}'Output \\
& Age_i \times w_{i,uft}'Output
\end{aligned}$
& Reject
& Reject
& $\boldsymbol{c}$
& 771{,}498 
& 3 
& 31, 31, 23 \\ [4.5em]

(vi)  
& $\begin{aligned}[t]
& w_{i,ut}'Peer, \\
& w_{i,uft}'Output \\
& Woman_i \times w_{i,uft}'Output, \\
& Age_i \times w_{i,uft}'Output
\end{aligned}$
& Reject
& Reject
& $\boldsymbol{c}$, $\boldsymbol{y}$
& 771{,}498 
& 4 
& 23, 23, 19 \\

\bottomrule
\end{tabular}
}
\justifying
\footnotesize
\textit{Note:} Each row corresponds to a specification in Table \ref{tab:peer_effects_ols} in Appendix \ref{app_C}. In Specification (i) the individual output is aggregated over all JEL
codes 
($n=42{,}861$), while in
specifications (ii)--(vi) individual output is defined at the JEL level ($n=771{,}498$). Columns 1--2 report the specification number and the construction of attribute peer exposure variables. The `Original result' column reports the conclusions reached in the original paper for the corresponding null hypothesis of no linear peer effects. The $\boldsymbol{s}$ test examines dependence operating through three rejection channels: the $\boldsymbol{c}$ channel through peers' attributes, the $\boldsymbol{y}$ channel through peers' outcomes, and the $\boldsymbol{cy}$ channel through both peers' attributes and outcomes. Column $n$ reports the sample size and column $l$ the number of peer attribute terms. The column `$p_c$, $p_y$, $p_{cy}$' reports the selected number of basis functions used in each test. Individual research output $y$ is the three-year moving average of articles divided by the number of co-authors, aggregated over fields in specification~(i) and measured at the JEL-code level thereafter. Peer output $w_{i,ut}'Output$ is the average of colleagues present in the department at date~$t$, where each colleague's productivity is their career-average publications per year; 
 from specification~(iii), this average is restricted to the same JEL code, $w_{i,uft}'Output$. 
Tests are performed at the 95\% confidence level. Standard errors are clustered at the department level.
\end{table}

\section{Conclusion}\label{Conclusions} 
Cross-unit dependence has long been recognized as a central concern in  economics, reflecting both fundamental identification problems and first-order implications for econometric inference \citep{manski1993reflection,Conley1999}. 
This paper proposes a novel nonparametric test for spillovers operating through peers’ attributes and/or outcomes, and provides a full asymptotic theory for it. The test has several appealing features. First, it can pick up nonlinear interactions of unknown form. Second, it only requires estimation under the null hypothesis of no spillovers, thereby avoiding nonparametric estimation altogether. Third, it is versatile, accommodating a wide range of data structures, including settings in which the interaction structure is incomplete or measured with error.

Our approach complements existing methods by offering a simple diagnostic to assess whether cross-unit dependence is present and whether linear approximations are likely to be informative. We illustrate its usefulness through four empirical applications, which suggest that the test can uncover forms of cross-unit dependence that are missed by standard specifications.
More broadly, our results reinforce the idea that the form and extent of spillovers should be informed by empirical evidence whenever possible.

\newpage

\bibliographystyle{chicago}
\bibliography{references}

\appendix
\renewcommand{\thesection}{A}
	\setcounter{equation}{0} \renewcommand{\theequation}{A.%
		\arabic{equation}} 
        \setcounter{table}{0} \renewcommand{\thetable}{A%
		\arabic{table}} 
	\setcounter{theorem}{0} \renewcommand{\thetheorem}{A\arabic{theorem}}

\section{Simulation study}\label{app:sims}

This section reports Monte Carlo results on the empirical size of our test. We set $k = 3$ and $\beta = (0.5, -2, 1)'$. The regressor matrix $X$ is $n\times k$ with ones in the first column and, in each replication, the remaining columns drawn independently across units from $U[-2,2]$ and $U[-2.5,2.5]$. Outcomes are generated by $y_i = x_i'\beta + \epsilon_i$, with no outcome or covariate spillovers. We implement the $\boldsymbol{cy}$ test using Hermite polynomials as the basis functions $\psi_j(\cdot)$ and (\ref{pxy_reco}) as the selection rule for $p_{cy}$, with $l=2$.

For the cluster-correlation specification paired with cluster-robust standard errors, stack $\epsilon = (\epsilon_{(1)}',\ldots,\epsilon_{(G)}')'$, where $\epsilon_{(g)}\in\mathbb{R}^{n_g}$. Then 
\begin{equation}
\label{eq:Var-epsilon-cluster}
var(\epsilon)
  = \Sigma
  = diag\left[ \Sigma_1,\ldots,\Sigma_G\right],
  \qquad
\Sigma_g =\frac{1}{2}I_{n_g}+\frac{1}{2} \mathbf{1}_{n_g}\mathbf{1}_{n_g}',
\end{equation}
where $\mathbf{1}_{n_g}$ is the $n_g\times 1$ vector of ones.
Hence $var(\epsilon_i)=1$, $cov(\epsilon_i,\epsilon_j)=1/2$ for distinct units in the same cluster, and $cov(\epsilon_i,\epsilon_j)=0$ across clusters. Draws satisfy $\epsilon = L \zeta $ with $L L' = \Sigma$ and $\zeta \in\mathbb{R}^n$ having independent $N(0,1)$ or $t_{10}$ components in the same cluster.\footnote{The $t_{10}$ draws are standardized to unit variance (i.e.\ multiplied by $\sqrt{(10-2)/10}$), so that $var(\zeta_i)=1$ and hence $var(\epsilon)=L L'=\Sigma$ holds.} components in the same cluster. Clusters partition indices into contiguous blocks: the first $G-1$ blocks have size $\lfloor n/G\rfloor$, and the $G$th block contains all remaining units.

We report empirical rejection frequencies under $H_0$ at nominal level $5\%$, with $1000$ replications, for $n =100,200,400,700,1000$ in non-lattice designs and $n =100,210,400,702,992$ on a two-dimensional lattice (see below for details). We report results with $\chi^2_q$ critical values as well as the asymptotic one based on $N(0,1)$. \cite{Guptaetal2024} observe that the former can control size better in smaller samples. Following \cite{Guptaetal2024}, we consider five designs for $W$, all normalized by spectral norm. 

\begin{enumerate}[label=(\arabic*)]
\item \textit{Exponential distance.}
\[
w_{ij} = \exp\left(-|s_i - s_j|\right)\,\mathbbm{1}\{|s_i - s_j| < \log n\},
\]
where $s_i$ is the location of unit $i$ on $[0,n]$, and $s_i \sim$i.i.d.\ $U[0,n]$.
\item \textit{Cutoff.}
\[
w_{ij} = \Phi(-d_{ij})\,\mathbbm{1}\{c_{ij} < n^{-2/3}\},
\]
with $\Phi$ the standard normal c.d.f., $d_{ij}\sim U[-3,3]$, $c_{ij}\sim U[0,1]$ i.i.d..
\item \textit{Circulant.}
$W_{i,i-1} = W_{i,i+1} = 1/2$, $i=1,\ldots,n$.
\item \textit{Random.}
$W$ is symmetric with entries in $\{0,1\}$ and a total of $\lfloor 2 n^{6/5}\rfloor$ non-zero elements off the diagonal. The average number of neighbors per row lies between $5.9$ and $9.2$ over the values of $n$ we use.
\item \textit{Lattice.}
Take $(m_1,m_2)$ with $n=m_1m_2$, e.g.\ $(10,10),(14,15),(20,20),(26,27),(31,32)$, yielding $n=100,210,400,702$ and $992$. Then generate 


\[
w_{m_2(k-1)+j,\, m_2(k-1)+j+1}
=
w_{m_2(k-1)+j,\, m_2(k-1)+j-1}
= 1
\]
for admissible $(k,j)$, and all other entries equal to zero.
\end{enumerate}

\begin{table}[h]
\centering
\caption{Monte Carlo size }\label{tab:sizesim}
\small
\renewcommand{\arraystretch}{1.25}
\begin{tabular}{ccccccccccccc}
\hline
& &  & \multicolumn{5}{c}{$\chi^2_q$ critical values} & \multicolumn{5}{c}{$N(0,1)$ critical value} \\ [.5 em]
$\zeta$ & $n$ & $p$ 
& expo & cutoff & circ & rand & latt
& expo & cutoff & circ & rand & latt \\ [.5 em]
\cmidrule(lr){1-3}\cmidrule(lr){4-8} \cmidrule(lr){9-13} 

$N(0,1)$ & 100  & 2 & 0.036 & 0.037 & 0.053 & 0.044 & 0.052 
          & 0.058 & 0.061 & 0.085 & 0.066 & 0.074 \\
& 200  & 3 & 0.047 & 0.037 & 0.084 & 0.042 & 0.050 
          & 0.071 & 0.059 & 0.109 & 0.065 & 0.081 \\
& 400  & 4 & 0.031 & 0.024 & 0.062 & 0.030 & 0.046 
          & 0.051 & 0.037 & 0.089 & 0.047 & 0.066 \\
& 700  & 4 & 0.039 & 0.037 & 0.059 & 0.047 & 0.047 
          & 0.057 & 0.053 & 0.084 & 0.072 & 0.066 \\
& 1000 & 5 & 0.030 & 0.038 & 0.052 & 0.022 & 0.049 
          & 0.045 & 0.052 & 0.076 & 0.042 & 0.071 \\ [1 em]

$t_{10}$ & 100  & 2 & 0.040 & 0.049 & 0.049 & 0.042 & 0.046
          & 0.062 & 0.078 & 0.081 & 0.070 & 0.075 \\
& 200  & 3 & 0.040 & 0.032 & 0.070 & 0.039 & 0.054
          & 0.062 & 0.057 & 0.099 & 0.072 & 0.087 \\
& 400  & 4 & 0.038 & 0.034 & 0.064 & 0.040 & 0.045
          & 0.060 & 0.058 & 0.097 & 0.055 & 0.067 \\
& 700  & 4 & 0.037 & 0.049 & 0.064 & 0.032 & 0.052
          & 0.052 & 0.062 & 0.088 & 0.051 & 0.074 \\
& 1000 & 5 & 0.021 & 0.044 & 0.068 & 0.028 & 0.043
          & 0.044 & 0.056 & 0.090 & 0.051 & 0.066 \\

\hline
\end{tabular}
\footnotesize
\justifying
\textit{Note:} Critical values are based on $\chi^2_q$ distributions and the standard normal distribution. Nominal size is $5\%$. All results are computed using cluster--robust variance. The data generating process  is simulated under the null hypothesis of no spillovers. The reported rejection frequencies correspond to the $\boldsymbol{cy}$ test. For the lattice design, the sample sizes are $n = 100, 210, 400, 702$ and $992$.
\end{table}

\noindent The results are reported in Table \ref{tab:sizesim}. For disturbances generated via $N(0,1)$ shocks, there appears to be some benefit from using the $\chi^2_q$ critical values even at $(n,p)=(1000,5)$ for the circulant and lattice cases. For the $t_{10}$ case, this remains true. In the remaining four designs for $W$, the results follow broadly the same pattern regardless of $N(0,1)$ or $t_{10}$ shocks: $\chi^2_q$ critical values can sometimes, though not always, control size better than the $N(0,1)$ one for small $(n,p)$, but this effect diminishes as $(n,p)$ grow. These results are on expected lines.

\renewcommand{\thesection}{B}
	\setcounter{equation}{0} \renewcommand{\theequation}{B.%
		\arabic{equation}} 
	\setcounter{theorem}{0} \renewcommand{\thetheorem}{B\arabic{theorem}}
  \section{Extensions}\label{sec:app_ext}
  \subsection{Heterogeneous cross-unit dependence}\label{sec:app_ext_mult}
  It is straightforward to extend our method to allow for multiple channels of cross-unit dependence, i.e. multiple social weight matrices, but we presented our theory for the single channel case in (\ref{model_1}) for notational simplicity.\footnote{Social interaction models with multiple social weight matrices have been characterized under a range of alternative assumptions \citep{Hsieh_Lin_2017,Arduini2020,comola_fortin_dieye}.} Indeed, suppose that we have $\ell$ channels of social dependence, each encoded in a weight matrix $W_h$, $h=1,\ldots,\ell$. Then we can write the model 
\begin{equation}\label{model_higherorder}
y_i=  \sum_{h=1}^{\ell}f_h\left(w_{h,i}' y\right) + x_i^\prime \beta +\sum_{h=1}^{\ell}\sum_{j=1}^l g_{hj}\left(w_{h,i}'c_j\right)+ \epsilon_i, \ \ \ i=1,....,n,
\end{equation}
where $f_h(\cdot)$ and $g_{hj}(\cdot)$ from $\mathbb{R}$ to $\mathbb{R}$ are $\ell(l+1)$ unknown functions and $W_h=\left(w_{h,1},\ldots,w_{h,n}\right)'$ are social weight matrices that are either fixed or exogenous and have zero diagonals, $h=1,\ldots,\ell$ and $j=1,\ldots,l$. The null of interest is then
\begin{equation}
	\mathcal{H}_0: f_h(s)=0 \text{ and } g_{hj}(s)=0,\;h=1,\ldots,\ell,\;j=1,\ldots,l,
\end{equation}
for all $s\in support(s)$. 

This can be approximated by series approximations exactly as in (\ref{lambda_prime}) and (\ref{approximate_null}) albeit with more subscripting, as we now show. Now we have the approximations 
\begin{equation}\label{lambda_prime_mult}
	f_h(s)=\sum_{i=1}^p \mu_{f_{h,i}} \psi_i(s)+r_{f_h}(s), g_{hj}(s)=\sum_{i=1}^p \mu_{g_{hj,i}} \psi_i(s)+r_{g_{hj}}(s),h=1,\ldots,\ell,j=1,\ldots,l,
\end{equation}
with $\mu_f= \left(\mu_{f_1,1},\ldots, \mu_{f_1,p},\ldots,\mu_{f_\ell,1},\ldots, \mu_{f_\ell,p}\right)^\prime$, $\mu_{g_j}= \left(\mu_{g_{1j},1},\ldots, \mu_{g_{1j},p},\mu_{g_{\ell j},1},\ldots, \mu_{g_{\ell j},p}\right)^\prime$  vectors of unknown series coefficients and $r_{f_h}(s)$, $r_{g_{hj}}(s)$ approximation errors. We define our approximate null hypothesis as
\begin{equation}\label{approximate_null_mult}
	\mathcal{H}_{0A}: \mu_f=0 \text{ and }\mu_{g_j}=0,\; j=1,\ldots,l, \text{ for some } \beta.
\end{equation} 
Now, define the $n\times 1$ vector $\Upsilon_{f_h,i}(y)= \left(\psi_i(w_{h,1}'y),\ldots,\psi_i(w_{h,n}'y)\right)'$
and the $n\times 1$ vectors $\Upsilon_{g_{hj},i}\left(c_j\right)= \left(\psi_i(w_{h,1}'c_j),\ldots,\psi_i(w_{h,n}'c_j)\right)'$, for each $i=1,\ldots,p$, $h=1,\ldots,\ell$ and $j=1,\ldots,l$. Next, concatenate these to write
$\Upsilon_{f,g}$ as before. This is now an $n\times q$ matrix, where $q=p\ell(l+1)$. We can now proceed as before in the $\ell=1$ case with these objects as the test statistic building blocks, noting that now we recommend
\begin{equation}\label{pxym_reco}
p_{cy,m}=\max\left\{2,\left[\frac{n^{1/3}}{\ell(l+1)}\right]\right\},  \boldsymbol {cy} \text{ test}.
\end{equation}

  \subsection{Embedded 
graphs}\label{sec:app_ext_netform}
We now consider a setting in which the observed social matrix is generated by an embedded graph model rather than treated as exogenous. By embedding the observed graph into a latent social space, we interpret it as a noisy measurement of an underlying latent structure. This reframing is particularly valuable in applied work, where social interaction data are often incomplete or measured with error, yet the latent structure remains informative to the researcher. When social interactions are mis-measured, our approach interprets these errors as small perturbations relative to the structurally generated network. We then derive formal conditions under which such perturbations are asymptotically negligible, ensuring the validity of inference based on the observed data.

Let us assume the researcher observes a network encoded by the $n\times n$ matrix $W(\kappa, z)$, the elements $w_{ij}(\kappa,z_o,z_l)$ of which are indicator functions that take the value unity if the unknown parameter vector $\kappa\in\mathscr{K}$, and the observed covariate vector $z_o$ and latent vector $z_l$ satisfy some prescribed condition. 
This specification encompasses standard network formation models with link functions that permit consistent estimation of $\kappa$, such as exponential link functions. The interpretation is that the adjacency matrix $W(\kappa,z)$ represents an embedding of the graph in a latent social space, rather than a fixed object observed without error. Such formulations generalize link formation based solely on observed covariates to settings where proximity in latent space governs tie formation, see for example \cite{Breza2020} and \cite{Lubold2023}. 

Accordingly we assume that, conditional on the latent variable vector $z_l$, the researcher has access to estimates $\hat{w}_{ij}=w_{ij}\left(\hat\kappa, z_o\right)$, where $\hat\kappa$ is some estimate of the parameterization of the link function that defines the probability of $w_{ij}=1$ as a function of $\kappa, z_o$ and $z_l$. Furthermore, we assume that there exists a sequence $s_n=s\rightarrow \infty$ such that
\begin{equation}\label{kappa_rate}
\left\Vert\hat\kappa-\kappa\right\Vert=O_p\left(s^{-1}\right),
\end{equation}
and that this rate of convergence carries over to the maximum row-sum of $ W\left(\hat\kappa,z\right)-W(\kappa,z)$, i.e.
\begin{equation}\label{row_sum_estimates}
\sup_{i=1,\ldots,n}\sum_{j=1}^n \left(\hat{w}_{ij}-w_{ij}\right)=O_p\left(s^{-1}\right).
\end{equation}

\noindent This condition is motivated by the observation that adjacency matrices require some control over their norms to limit dependence to a manageable degree. Row and column summability is a typical assumption. Indeed, if $w_{ij}(\kappa,z)$ were a differentiable function in $\kappa$ we could use the mean value theorem to write $\hat{w}_{ij}-w_{ij}=\left.\frac{\partial w_{ij}(\kappa,z)}{\partial \kappa}\right\vert_{\kappa=\bar\kappa}'\left(\hat{\kappa}-\kappa\right)$ for an intermediate point $\bar\kappa$ and obtain (\ref{row_sum_estimates}) if we assume row-summability of the derivative matrix, i.e. 
\[
\sup_{i=1,\ldots,n}\sup_{\kappa\in\mathscr{K}}\sum_{j=1}^n\frac{\partial w_{ij}(\kappa,z)}{\partial \kappa}=O_p(1),
\]
uniformly in $z$. Of course in our case $\hat w_{ij}$ and $w_{ij}$ are non-differentiable because they are indicator functions, hence the condition (\ref{row_sum_estimates}). For example, if each unit only has a fixed number of neighbours, as in a `$k$ nearest neighbours' setup, then (\ref{row_sum_estimates}) will be satisfied as long as 
\begin{equation}\label{w_diff}
\hat{w}_{ij}-w_{ij}=O_p\left(s^{-1}\right),
\end{equation}
because the sum on the LHS of (\ref{row_sum_estimates}) will have only a fixed number of non-zero summands.

Now, for each $i=1,\ldots,p$, define the $n\times 1$ vector $\hat\Upsilon_{f,i}(y)= \left(\psi_i(\hat w_1'y),\ldots,\psi_i(\hat w_n'y)\right)'$
and the $n\times 1$ vectors $\hat\Upsilon_{g_j,i}(c_j)= \left(\psi_i(\hat w_1'c_j),\ldots,\psi_i(\hat w_n'c_j)\right)'$, where $\hat w_i$ has elements $\hat w_{ij}$, $j=1,\ldots,n$, and write
\[
\hat\Upsilon_{f,g}=\begin{pmatrix} \hat\Upsilon_{f,1}(y)& \ldots& \hat\Upsilon_{f,p}(y)&\hat\Upsilon_{g_1,1}(c_1)& \ldots& \hat\Upsilon_{g_1,p}(c_1)&\ldots& \hat\Upsilon_{g_l,1}(c_l)& \ldots& \hat\Upsilon_{g_l,p}(c_l) \end{pmatrix},
\]
which is an $n\times q$ matrix, where $q\sim p$ asymptotically. Then, we can apply our test as long as the estimates $\hat{w}_{ij}$ satisfy
\begin{equation}\label{Upsilonapprox}
\left\Vert n^{-1/2}\hat\Upsilon_{f,g}-n^{-1/2}\Upsilon_{f,g}\right\Vert=o_p(1),    
\end{equation}
conditional on $z_l$. The squared LHS of (\ref{Upsilonapprox}) is bounded by a sum of $nq$ terms of the type
\begin{equation}\label{Upsilonapprox2}
n^{-1}\left(\psi\left(\hat w'\ell\right)-\psi\left(w'\ell\right)\right)^2,    
\end{equation}
where we omit subscripting for brevity and let $\ell$ denote a generic observed $n\times 1$ vector with components $\ell_i$ such that $\ell_i=O_p(1)$, uniformly in $i$. Assuming that the generic basis function $\psi(\cdot)$ is differentiable with derivative $\psi'(\cdot)$ and $E(\psi'(x))^2<C$, with $\boldsymbol{c}$ a generic constant, we use the mean value theorem and (\ref{row_sum_estimates}) to observe that (\ref{Upsilonapprox2}) is
\begin{eqnarray}
n^{-1}\psi'(\bar{x})^2\left(\left(\hat w-w\right)'\ell\right)^2&=&O_p\left(n^{-1}\right)\cdot\left(\sum_{i=1}^n\left(\hat w_i-w_i\right)\ell_i\right)^2\nonumber\\
&=&O_p\left(n^{-1}\right)\cdot\left(\sum_{i=1}^n\left(\hat w_i-w_i\right)\right)^2\nonumber\\
&=&
O_p\left(s^{-2}n^{-1}\right),\label{Upsilonapprox3}
\end{eqnarray}
 where $\hat w'\ell\leq \bar x\leq w'\ell$ and $\hat w_i$ and $w_i$ are elements of the $n\times 1$ vectors $\hat w$ and $w$, respectively. Then we conclude that
\begin{equation}\label{Upsilonapprox4}
\left\Vert n^{-1/2}\hat\Upsilon_{f,g}-n^{-1/2}\Upsilon_{f,g}\right\Vert^2=O_p\left(ps^{-2}\right),    
\end{equation}
so that (\ref{Upsilonapprox}) holds if $p^{1/2}s^{-1}\rightarrow 0$.   \renewcommand{\thesection}{C}
	\setcounter{equation}{0} \renewcommand{\theequation}{C.%
		\arabic{equation}} 
	\setcounter{theorem}{0} \renewcommand{\thetheorem}{C\arabic{theorem}}
    
\section{Proofs of Theorems}\label{appendix:proofs}
\subsection{Preliminary results}
\allowdisplaybreaks
	\begin{theorem}\label{theorem:dhatd} 
		Under Assumptions \ref{ass:errors}-\ref{ass:Mhat}, under $\mathcal{H}_{0A}$ in (\ref{approximate_null}), for $p^3/n \rightarrow 0$ as $n\rightarrow \infty$,
		\begin{equation}
			\left\Vert \hat{d} - d \right \Vert =O_p\left( \frac{p^{3/2}}{n} \right).
		\end{equation}
		
	\end{theorem}
	
	\begin{proof}
    We first establish a preliminary bound. Let $1_g(i,j)$ be and indicator function that takes the value 1 when $i$ and $j$ are in the same cluster $g$ and zero otherwise. Observe that, by Assumptions \ref{ass:errors}, \ref{ass:eigsandinsts} and $m\sim p$,
    \begin{eqnarray}
			\mathbb{E}\left\Vert n^{-1}Z^\prime \epsilon\right\Vert^2&=&n^{-2}\sum_{i=1}^n\sigma_i^2\mathbb{E}\left\Vert z_i\right\Vert^2+n^{-2}\sum_{g=1}^G\sum_{i\neq j}1_g(i,j)\sigma_{ij}\mathbb{E}z_{ig}'z_{jg}\nonumber\\
            &\leq& Kn^{-1}m+n^{-2}\sum_{g=1}^G\sum_{i\neq j}1_g(i,j)\sigma_{ij}\left(\mathbb{E}\left\Vert z_{ig}\right\Vert^2\right)^{1/2}\left(\mathbb{E}\left\Vert z_{jg}\right\Vert^2\right)^{1/2}\nonumber\\
            &=& O\left(p\left(n^{-1}+n^{-2}\sum_{g=1}^Gn_g^2\right)\right)=O\left(pn^{-1}\right),\label{theorem1_1}
		\end{eqnarray}
	where the last equality follows because  $n^{-2}\sum_{g=1}^Gn_g^2=O\left(n^{-2}G\right)=O\left(n^{-1}\right)$, recalling that $\sup_{g=1,\ldots,G}n_g<K$ and so $G$ and $n$ have the same asymptotic order. Thus, by the Markov inequality,
\begin{equation}\label{Zprimeespilonbound}
\left\Vert n^{-1}Z^\prime \epsilon\right\Vert=O_p\left(\sqrt{\frac{{p}}{n}}\right).
\end{equation}
We also note that, under Assumptions \ref{ass:regressors}, \ref{ass:eigsandinsts} and \ref{ass:Mhat},
		\begin{align}\label{IV_rate1}
			\left \Vert \hat{\beta}- \beta\right \Vert =& \left\Vert\left(\frac{1}{n}X^\prime \mathcal{P}_Z X\right)^{-1} \frac{1}{n} X^\prime \mathcal{P}_Z \epsilon \right\Vert =O_p\left( \left\Vert \frac{Z^\prime \epsilon}{n}\right \Vert\right)  .
		\end{align}
			
		\noindent Let $R=\left( r(w_1' y), \ldots ,r(w_n' y)  \right)^\prime$ be the $n\times 1$ vector of approximation errors in (\ref{lambda_prime}) with $R_i=r(w_i'y)$. From the 2SLS expression for $\hat{\beta}- \beta$ in (\ref{IV_rate1}),
		\begin{align}\label{dhat_equiv}
			\hat{d}=& -\frac{2}{n}U^\prime  \mathcal{P}_Z \left(I -  X  ( X ^\prime \mathcal{P}_Z  X )^{-1}  X ^\prime \mathcal{P}_Z \right) \epsilon  -\frac{2}{n}U^\prime  \mathcal{P}_Z R \notag \\
			=& -\frac{2}{n}U^\prime \mathcal{P}_Z \left(I - \mathcal{P}_Z  X  ( X ^\prime \mathcal{P}_Z  X )^{-1}  X ^\prime \mathcal{P}_Z \right) \mathcal{P}_Z \epsilon  -\frac{2}{n} U^\prime  \mathcal{P}_Z R \notag \\
			=& -\frac{2}{n} \hat{J}^\prime \hat{M}^{-1/2}\left(I - \hat{M}^{-1/2} \hat{N}\left( \hat{N}^\prime \hat{M}^{-1} \hat{N} \right)^{-1} \hat{N}^\prime \hat{M}^{-1/2} \right) \hat{M}^{-1/2} Z^\prime \epsilon - \frac{2}{n} \hat{J}^\prime \hat{M}^{-1} Z^\prime R \notag \\
			=& -\frac{2}{n} \hat{J}^\prime \hat{M}^{-1/2} \hat{\mathcal{K}}_{NM} \hat{M}^{-1/2} Z^\prime \epsilon - \frac{2}{n} \hat{J}^\prime \hat{M}^{-1} Z^\prime R, 
		\end{align}
		where $\hat{\mathcal{K}}_{NM}= \left(I - \hat{M}^{-1/2} \hat{N}\left( \hat{N}^\prime \hat{M}^{-1} \hat{N} \right)^{-1} \hat{N}^\prime \hat{M}^{-1/2} \right)$. 		From (\ref{dhat}), we write
		\begin{align}\label{d_decomposition}
			\left\Vert \hat{d} - d\right \Vert  \leq \left\Vert \frac{2}{n} \hat{J}^\prime \hat{M}^{-1/2}\hat{\mathcal{K}}_{NM}\hat{M}^{-1/2} Z^\prime \epsilon  - \frac{2}{n} J^\prime M^{-1/2}\mathcal{K}_{NM} M^{-1/2} Z^\prime \epsilon  \right \Vert +\left\Vert \frac{2}{n} \hat{J}^\prime \hat{M}^{-1} Z^\prime R\right\Vert.
		\end{align}
		Subsequently, denote $\Delta^A_B=A-B$ for conformable $A$ and $B$. Then, via some standard albeit tedious algebra, the first term on the RHS of (\ref{d_decomposition}) is bounded by
		\begin{align}
			&\left\Vert \Delta^{\hat{J}}_J \right\Vert \left\Vert \hat{M}^{-1} \right\Vert \left\Vert \frac{1}{n} Z^\prime \epsilon \right\Vert + \left\Vert J\right\Vert \left\Vert \Delta^{\hat{M}^{-1}}_{M^{-1}}\right \Vert \left\Vert \frac{1}{n}Z^\prime \epsilon \right\Vert  \notag\\
			&+\left\Vert \Delta^{\hat{J}}_{J}\right \Vert \left\Vert \hat{M}^{-1}\right \Vert \left\Vert \hat{N}\right \Vert \left\Vert \left( \hat{N}^\prime \hat{M}^{-1} \hat{N}\right)^{-1}\right \Vert \left\Vert \hat{N} \right \Vert  \left\Vert \hat{M}^{-1} \right \Vert \left\Vert \frac{1}{n}Z^\prime \epsilon\right \Vert \notag \\
			+& \left\Vert J \right \Vert \left\Vert\Delta^{\hat{M}^{-1}}_{M^{-1}}\right \Vert \left\Vert \hat{N}\right \Vert \left\Vert \left( \hat{N}^\prime \hat{M}^{-1} \hat{N}\right)^{-1}\right \Vert \left\Vert \hat{N} \right \Vert  \left\Vert \hat{M}^{-1}\right \Vert \left\Vert \frac{1}{n}Z^\prime \epsilon\right \Vert \notag \\ +& \left\Vert J \right \Vert \left\Vert M^{-1}\right \Vert \left\Vert \Delta^{\hat{N}}_{N}\right \Vert \left\Vert \left( \hat{N}^\prime \hat{M}^{-1} \hat{N}\right)^{-1}\right \Vert \left\Vert \hat{N} \right \Vert  \left\Vert \hat{M}^{-1}\right \Vert \left\Vert \frac{1}{n}Z^\prime \epsilon\right \Vert  \notag \\
			+& \left\Vert J \right \Vert \left\Vert M^{-1}\right \Vert \left\Vert N\right \Vert \left\Vert \Delta^{\left( \hat{N}^\prime \hat{M}^{-1} \hat{N}\right)^{-1}}_{\left(N^\prime M^{-1} N\right)^{-1}} \right \Vert \left\Vert \hat{N} \right \Vert  \left\Vert \hat{M}^{-1}\right \Vert \left\Vert \frac{1}{n}Z^\prime \epsilon\right \Vert \notag \\
			+& \left\Vert J \right \Vert \left\Vert M^{-1}\right \Vert \left\Vert N\right \Vert \left\Vert  \left(N^\prime M^{-1} N\right)^{-1} \right \Vert \left\Vert\Delta^{\hat{N}}_{N}  \right \Vert  \left\Vert \hat{M}^{-1}\right \Vert \left\Vert \frac{1}{n}Z^\prime \epsilon\right \Vert \notag \\+& \left\Vert J \right \Vert \left\Vert M^{-1}\right \Vert \left\Vert N\right \Vert \left\Vert  \left(N^\prime M^{-1} N\right)^{-1} \right \Vert \left\Vert N  \right \Vert  \left\Vert \Delta^{\hat{M}^{-1}}_{M^{-1}}\right \Vert \left\Vert \frac{1}{n}Z^\prime \epsilon\right \Vert 
		\end{align}
		Under Assumption \ref{ass:Mhat}, we have
		\begin{equation}
			\left\Vert \Delta^{\hat{N}}_N\right \Vert =O_p\left(\frac{p}{\sqrt{n}}\right) \ \ \ \ \text{and} \ \ \ \ \  \left\Vert \Delta^{\hat{J}}_J \right \Vert =O_p\left(\frac{p}{\sqrt{n}}\right) .
		\end{equation}
		Also, under Assumptions \ref{ass:eigsandinsts} and \ref{ass:Mhat}, 
		\begin{equation}
			\left\Vert\Delta^{\hat{M}^{-1}}_{M^{-1}} \right \Vert \leq \left\Vert M^{-1}\right \Vert  \left\Vert \hat{M}^{-1}\right \Vert  \left\Vert\Delta^{\hat{M}}_{M} \right \Vert = O_p\left(\frac{p}{\sqrt{n}}\right)
		\end{equation}
		and similarly, under Assumptions  \ref{ass:eigsandinsts} and \ref{ass:Mhat}, 
		\begin{align}
			\left\Vert \Delta^{\left( \hat{N}^\prime \hat{M}^{-1} \hat{N}\right)^{-1}}_{\left(N^\prime M^{-1} N\right)^{-1}}\right\Vert &\leq \left\Vert N^\prime M^{-1} N \right \Vert  \left\Vert \hat{N}^\prime \hat{M}^{-1} \hat{N}\right \Vert \left\Vert \Delta^{ \hat{N}^\prime \hat{M}^{-1} \hat{N}}_{N^\prime M^{-1} N}\right \Vert \notag\\
			&= O_p\left(\frac{p}{\sqrt{n}}\right).
		\end{align}
		
		\noindent Thus, upon recalling (\ref{Zprimeespilonbound}), the first term at the RHS of (\ref{d_decomposition}) is observed to be $O_p\left( n^{-1}{p^{3/2}} \right)$. The second term at the RHS of (\ref{d_decomposition}) is instead
		\begin{equation}\label{d_remainder}
			\left\Vert \frac{2}{n}\hat{J}^\prime \hat{M}^{-1} Z^\prime R \right \Vert = O_p\left(\frac{1}{n} \left\Vert R\right \Vert \left\Vert Z\right \Vert \right) = O_p(p^{-\nu}),
		\end{equation}
		where the first equality at the RHS of (\ref{d_remainder}) follows under Assumptions \ref{ass:eigsandinsts} and \ref{ass:Mhat}. The second equality follows since $\Vert Z \Vert = O_p(\sqrt{n})$ under Assumptions  \ref{ass:eigsandinsts} and \ref{ass:Mhat}, and each component of the $n\times 1$ vector $R$ is $	 O_p(p^{-\nu})$
		by Assumption \ref{ass:eigsandinsts}, and hence $||R||= O_p(\sqrt{n}  p^{-\nu})$. The last equality in (\ref{d_remainder}) follows from Assumption \ref{ass:eigsandinsts}.
		
		Under Assumption \ref{ass:Mhat}, the first term in (\ref{d_decomposition}) dominates the second one as long as $\nu$ satisfies ${n}/{p^{\nu+3/2}}=o(1)$ as $n\rightarrow \infty$, which holds under Assumption \ref{ass:eigsandinsts}.
	\end{proof}

\begin{theorem}\label{theorem:approx}
	Under Assumptions \ref{ass:errors}-\ref{ass:Mhat}, $\nu>5/2$, under $\mathcal{H}_{0A}$ in (\ref{approximate_null}) and $p^3/n =o(1)$,
	\begin{equation}
		\mathcal{S} - \frac{n d^\prime H^{-1} d - q}{\sqrt{2q}} = o_p(1), \text{ as } n\rightarrow \infty. 
	\end{equation}
\end{theorem}

\begin{proof}
We can equivalently prove
		\begin{equation}\label{th2_equiv}
			\hat{d}^\prime \hat{H}^{-1}\hat{d} -d^\prime H^{-1}d= o_p\left(\frac{\sqrt{p}}{n}\right).
		\end{equation}
		Write the LHS of (\ref{th2_equiv}) as
		\begin{equation}\label{th2_terms}
			\left(\hat{d} -d \right)^\prime \hat{H}^{-1}\hat{d} + d^\prime H^{-1} (\hat{d}-d) + d^\prime \hat{H}^{-1}\left(H- \hat{H}\right) H^{-1}\hat{d},
		\end{equation}
		which has norm bounded by
		\begin{equation}\label{theorem2_1}
			K\left\Vert \Delta^{\hat{d}}_d \right \Vert \left\Vert \hat{H}^{-1} \right \Vert \left\Vert \hat{d}\right \Vert + K\left\Vert\Delta^{\hat{d}}_d \right \Vert \left\Vert H^{-1} \right \Vert \left\Vert d\right \Vert + K \left\Vert d\right \Vert  \left \Vert  \hat{H}^{-1} \right \Vert \left \Vert\Delta^{\hat{H}}_H\right \Vert \left \Vert H^{-1}\right \Vert \left\Vert \hat{d}\right \Vert .
		\end{equation}
		From Theorem \ref{theorem:dhatd}, $\left\Vert \Delta^{\hat{d}}_d \right \Vert = O_p\left( n^{-1}{p^{3/2}} \right)$. 
		Under Assumptions \ref{ass:eigsandinsts} and \ref{ass:Mhat}, and from (\ref{theorem1_1}) we have $\left\Vert d\right \Vert = O_p\left({\sqrt{p/{n}}}\right)$. Also, from Theorem \ref{theorem:dhatd},
		\begin{equation}
			\left\Vert \hat{d}\right \Vert \leq \left\Vert \Delta^{\hat{d}}_d\right \Vert + \left\Vert d\right \Vert = O_p\left(\sqrt{\frac{p}{n}}\right),
		\end{equation}
		where the last equality is due to $p^2/n =o(1)$.
		Also, under Assumptions \ref{ass:errors},  \ref{ass:eigsandinsts} and \ref{ass:Mhat}, $\left\Vert \hat{H}^{-1} \right \Vert=O_p(1)$ and $\left\Vert H^{-1} \right \Vert=O_p(1)$. Thus,  the first and second terms in (\ref{theorem2_1}) are $O_p\left({p^2}/{n^{3/2}}\right)$, and these are $o_p\left(p^{1/2}/n\right)$ if ${p^{3}}/{n}=o(1)$.
		
		Using 2SLS estimates for $\beta_0$ and proceeding as in (\ref{dhat_equiv}), we can write
		\begin{align}\label{Hhat_equiv}
			\hat{H}&= 4 \hat{J}^\prime \hat{M}^{-1/2} \hat{\mathcal{K}}_{NM} \hat{M}^{-1/2} \tilde{\Phi}\hat{M}^{-1/2} \hat{\mathcal{K}}_{NM} \hat{M}^{-1/2} \hat{J} + \frac{4}{n} \hat{J}^\prime \hat{M}^{-1/2}\hat{\mathcal{K}}_{NM} \hat{M}^{-1/2} Z^\prime \epsilon R^\prime Z \hat{M}^{-1} \hat{J} \notag \\
			&+ \frac{4}{n} \left(\hat{J}^\prime \hat{M}^{-1/2}\hat{\mathcal{K}}_{NM} \hat{M}^{-1/2} Z^\prime \epsilon R^\prime Z \hat{M}^{-1} \hat{J}\right)'+\frac{4}{n} \hat{J}^\prime \hat{M}^{-1} Z^\prime R R^\prime Z \hat{M}^{-1} \hat{J}
		\end{align}
		where $\tilde{\Phi}=Z^\prime \tilde{\Sigma} Z/n$, with $\tilde{\Sigma}$ being an $n\times n$  block-diagonal matrix such that its $g$-th block $\tilde{\Sigma}_g$ has elements   $\tilde{\Sigma}_{gij}=\epsilon_{ig}\epsilon_{jg}$. From (\ref{H}), we write
		\begin{align}\label{H_conv}
			\left \Vert \Delta^{\hat{H}}_H \right \Vert  \leq & \left\Vert \hat{J}^\prime \hat{M}^{-1/2} \hat{\mathcal{K}}_{NM} \hat{M}^{-1/2} \tilde{\Phi}\hat{M}^{-1/2} \hat{\mathcal{K}}_{NM} \hat{M}^{-1/2} \hat{J} \right.\notag\\
			& \left.- J^\prime M^{-1/2} \mathcal{K}_{NM} M^{-1/2} \Phi M^{-1/2}\mathcal{K}_{NM} M^{-1/2} J \right\Vert \notag \\
			+& 2\left\Vert \frac{4}{n} \hat{J}^\prime \hat{M}^{-1/2}\hat{\mathcal{K}}_{NM} \hat{M}^{-1/2} Z^\prime \epsilon R^\prime Z \hat{M}^{-1} \hat{J} \right\Vert + \left\Vert \frac{4}{n} \hat{J}^\prime \hat{M}^{-1} Z^\prime R R^\prime Z \hat{M}^{-1} \hat{J}
			\right \Vert.
		\end{align}
		By standard algebra, the first term in (\ref{H_conv}) is bounded by
		\begin{align}
			\begin{split}
				&\left\Vert \hat{J}^\prime \hat{M}^{-1} \tilde{\Phi} \hat{M}^{-1} \hat{J} -  J^\prime M^{-1} \Phi M^{-1} J \right\Vert \notag \\+ &\left\Vert \hat{J}^\prime \hat{M}^{-1} \hat{N} \left(\hat{N}^\prime \hat{M}^{-1} \hat{N}\right)^{-1} \hat{N}^\prime \hat{M}^{-1} \tilde{\Phi}\hat{M}^{-1}\hat{J} - J^\prime M^{-1} N \left(N^\prime M^{-1} N\right)^{-1} N^\prime M^{-1} \Phi M^{-1} J \right\Vert \notag \\
				+& \left\Vert \hat{J}^\prime \hat{M}^{-1} \hat{N} \left(\hat{N}^\prime \hat{M}^{-1} \hat{N}\right)^{-1} \hat{N}^\prime \hat{M}^{-1} \tilde{\Phi}\hat{M}^{-1} \hat{N} \left(\hat{N}^\prime \hat{M}^{-1} \hat{N} \right)^{-1} \hat{N} \hat{M}^{-1}\hat{J} \right. \notag \\- & \left.  J^\prime M^{-1} N \left(N^\prime M^{-1} N\right)^{-1} N^\prime M^{-1} \Phi M^{-1} N \left(N^\prime M^{-1} N \right)^{-1} N M^{-1}J \right \Vert .
			\end{split}
		\end{align}
		\normalsize
		We provide details for the first term in the last displayed expression, the others following similarly. Specifically,
		\begin{align}\label{H_conv_first_1}
			&\left\Vert \hat{J}^\prime \hat{M}^{-1} \tilde{\Phi} \hat{M}^{-1} \hat{J} -  J^\prime M^{-1} \Phi M^{-1} J \right\Vert  \leq  \left \Vert\Delta^ {\hat{J}}_J \right\Vert \left\Vert \hat{M}^{-1}\right\Vert ^2 \left\Vert \tilde{\Phi}\right\Vert  \left\Vert \hat{J}\right\Vert \notag \\+& \left \Vert J \right\Vert \left\Vert \Delta^{\hat{M}^{-1}}_{M^{-1}}\right\Vert \left\Vert \tilde{\Phi}\right\Vert \left\Vert \hat{M}^{-1}\right\Vert \left\Vert \hat{J}\right\Vert + \left \Vert J \right\Vert \left\Vert M^{-1}\right\Vert \left\Vert \Delta^{\tilde{\Phi}}_\Phi \right\Vert \left\Vert \hat{M}^{-1}\right\Vert \left\Vert \hat{J}\right\Vert \notag \\
			+& \left \Vert J \right\Vert \left\Vert M^{-1}\right\Vert \left\Vert \Phi\right\Vert \left\Vert \Delta^{\hat{M}^{-1}}_{M^{-1}}\right\Vert \left\Vert \hat{J}\right\Vert + \left \Vert J \right\Vert \left\Vert M^{-1}\right\Vert^2 \left\Vert \Phi\right\Vert \left\Vert \Delta^{\hat{J}}_{J}\right\Vert.
		\end{align}
		Under Assumptions  \ref{ass:eigsandinsts} and \ref{ass:Mhat}, most terms can be handled as in the proof of Theorem \ref{theorem:dhatd}, and $\left\Vert \Delta^{\hat{M}^{-1}}_{M^{-1}}\right\Vert=O_p(p/\sqrt{n})$ and $\left\Vert \Delta^{\hat{J}}_{J}\right\Vert=O_p(p/\sqrt{n})$. Focusing instead on $\left\Vert \Delta^{\tilde{\Phi}}_\Phi \right \Vert$, observe that
		\begin{equation}\label{omega_terms}
			\left\Vert \Delta^{\tilde{\Phi}}_\Phi \right \Vert \leq \left\Vert \Delta^{\tilde{\Phi}}_{\bar{\Phi}}\right\Vert + \left\Vert\Delta^{\bar{\Phi}}_{\Phi} \right\Vert,
		\end{equation}
		with $\bar{\Phi}= Z^\prime \Sigma Z /n$. The first term in (\ref{omega_terms}) is $\left\Vert Z^\prime \Delta^{\tilde{\Sigma}}_\Sigma Z/n \right\Vert$, where the $m\times m$ matrix $Z^\prime \Delta^{\tilde{\Sigma}}_\Sigma Z/n$ has typical element $n^{-1}\sum_{i,j=1}^n\sum_{g=1}^Gt_{ijrsg}\left(\epsilon_{ig}\epsilon_{jg}-\sigma_{ijg}\right)$ where we write $t_{ijrsg}=1_g(i,j)z_{irg}z_{jsg}$. This typical element has zero mean and variance
		\begin{align}
			2n^{-2}\sum_{i,j=1}^n\sum_{g=1}^G t^2_{ijrsg}\sigma^2_{ig}\sigma^2_{jg}+n^{-2}\sum_{i=1}^n\sum_{g=1}^G t^2_{iirsg}\left (\mathbb{E}\epsilon^4_{ig}-3\sigma^4_{ig}\right)=O\left(n^{-2}\sum_{g=1}^G n^2_g\right)=O\left(n^{-1}\right),
		\end{align}
		under Assumptions \ref{ass:errors} and \ref{ass:eigsandinsts} and, since, $m\sim p$ we therefore $ \left\Vert \Delta^{\tilde{\Phi}}_{\bar{\Phi}}\right\Vert =O_p(p/\sqrt{n})$.
		The matrix in the norm in the second term in (\ref{omega_terms}) has mean zero and the second moment of its squared Euclidean norm is bounded by 
$K\xi m^2/n^2= O\left({p^2}/{n}\right)$
		under Assumptions  \ref{ass:errors}, \ref{ass:eigsandinsts} and \ref{ass:Mhat}, rendering it $O_p(p/\sqrt{n})$ and thus $\left\Vert \Delta^{\tilde{\Phi}}_{{\Phi}}\right\Vert= O_p(p/\sqrt{n})$. 
		We then conclude (\ref{H_conv_first_1}) is $O_p\left(p/\sqrt{n}\right)$. Similar steps yield that the first term on the RHS of  (\ref{H_conv}) is $O_p\left(p/\sqrt{n}\right)$. 
		
		By similar arguments to those that led to (\ref{d_remainder}), under Assumptions \ref{ass:eigsandinsts} and \ref{ass:Mhat}, the second term in (\ref{H_conv}) is bounded by
		\begin{equation}
			K \left\Vert \frac{1}{n}Z^\prime \epsilon \right\Vert \left\Vert Z\right\Vert \left\Vert R\right\Vert =O_p\left(\frac{\sqrt{n}}{p^{\nu-1/2}} \right),
		\end{equation}
		which is negligible compared to the first term in (\ref{H_conv}) since $n/p^{(\nu+1/2)}=o(1)$, under Assumptions \ref{ass:eigsandinsts} and  \ref{ass:Mhat}.
		Similarly, the third term in (\ref{H_conv}) is $O_p(n p^{-2\nu})$, which is negligible compared to the first term since $n^{3/2}/{p^{2\nu+ 1}}=o(1)$ as $n\rightarrow \infty$, under Assumptions \ref{ass:eigsandinsts} and \ref{ass:Mhat}. 
		We conclude that
		\begin{equation}
			\left\Vert \Delta^{\hat{H}}_H\right\Vert = O_p\left(p/\sqrt{n}\right).
		\end{equation}
		By Assumption \ref{ass:Mhat}, the last term in (\ref{theorem2_1}) is thus $O_p(p^2/n^{3/2})$, given $\Vert d\Vert=O_p(\sqrt{p/n})$ and  $\left\Vert\hat{d}\right\Vert=O_p(\sqrt{p/n})$. Hence, the second term in (\ref{theorem2_1}) is $o_p(\sqrt{p}/n)$ as long as $p^3/n=o(1)$, concluding the proof.

\end{proof}
\subsection{Proofs of main theorems}
\begin{proof}[Proof of Theorem \ref{theorem:nulldist}:]
 We have $\mathbb{E}\epsilon_{ig}\epsilon_{jg'}=0$ for $g\neq g'$ and $\mathbb{E}\epsilon_{ig}\epsilon_{jg}=\sum_{r=1}^{n_g}b_{ir}b_{jr}$. Thus $\sigma_{ijg}=\sum_{r=1}^{n_g}b_{ir}b_{jr}$ for $i\neq j$ and $\sigma^2_{ig}=\sum_{r=1}^{n_g}b^2_{ir}$. Because $n_g$ are finite and fixed, clearly only a finite number of the $b_{ir}$ are non-zero for any given $i$ or $r$. Thus we have $\sup_{r=1,\ldots,n}\sum_{i=1}^n \left\vert b_{ir}\right\vert+\sup_{i=1,\ldots,n}\sum_{r=1}^n \left\vert b_{ir}\right\vert<K$. Upon writing $\epsilon=B\eta$, where $\eta$ is an $n\times 1$ vector with elements $\eta_r$ and $B$ is an $n\times n$ matrix with elements $b_{ir}$, we observe that 
\begin{equation}\label{d_linpro}
d=-\frac{2}{n}J'M^{-1/2}\mathcal{K}_{NM}M^{-1/2}Z'B\eta.
\end{equation}
In view of Theorem \ref{theorem:approx}, we know that $\mathcal{S} - \frac{n d^\prime H^{-1} d - q}{\sqrt{2q}} = o_p(1)$ and therefore it is sufficient to show that $\frac{n d^\prime H^{-1} d - q}{\sqrt{2q}}\overset{d}{\rightarrow} N(0,1)$, which by (\ref{d_linpro}) boils down to showing that
\begin{equation}\label{clt_target}
\frac{n \eta' B'\mathscr{G}B \eta - q}{\sqrt{2q}}\overset{d}{\rightarrow} N(0,1),
\end{equation}
where 
\[
\mathscr{G}=\frac{4}{n^2}ZM^{-1/2}\mathcal{K}_{NM}M^{-1/2}JH^{-1}J'M^{-1/2}\mathcal{K}_{NM}M^{-1/2}Z'=\frac{4}{n}Z\mathscr{A}Z', 
\]
say, where $\mathscr{A}=n^{-1}M^{-1/2}\mathcal{K}_{NM}M^{-1/2}JH^{-1}J'M^{-1/2}\mathcal{K}_{NM}M^{-1/2}$. Theorem A.1 of \cite{Guptaetal2025} applies if
\begin{equation}\label{clt_eigs}
\eigbig{\mathscr{G}}=O_p(i) \text{ and } \left(\eigsmall{\mathscr{G}}\right)^{-1}=O_p(1),
\end{equation}
and
\begin{equation}\label{clt_sum1}
g_{ij}=O_p(p/n)\text{ and } \sum_{i=1}^n g^2_{ij}=O_p(p/n),
\end{equation}
uniformly in $i$ and $j$. The conditions in Assumption \ref{ass:eigsandinsts} ensure that (\ref{clt_eigs}) holds. To check (\ref{clt_sum1}),
observe that 
\[
g_{ij}=\frac{4}{n}z_i'\mathscr{A}z_j=O_p\left({{\left\Vert z_i\right\Vert}{\left\Vert z_j\right\Vert}}/{n}\right)=O_p(p/n)
\]
and  
\[
\sum_{j=1}^ng^2_{ij}=\frac{16}{n^2}z_i'\mathscr{A}\left(\sum_{j=1}^n z_jz_j'\right)\mathscr{A}z_i=\frac{16}{n}z_i'\mathscr{A}\hat{M}\mathscr{A}z_i=O_p\left({{\left\Vert z_i\right\Vert}{\left\Vert z_j\right\Vert}}/{n}\right)=O_p(p/n),
\]
as desired. Then (\ref{clt_target}) follows by Theorem A.1 of \cite{Guptaetal2025}.
\end{proof}

\begin{proof}[Proof of Theorem \ref{theorem:consistency}:]

Let $\gamma=(\mu^\prime, \beta^\prime)^\prime$. Corresponding to $\tilde{d}= \partial \mathcal{Q}/\partial \gamma$ defined in (\ref{d}) under $\mathcal{H}_{0A}$,  we now define the unconstrained gradient vector, $(q+k)\times 1$, $\tilde{d}_{U}$ as 
		\begin{equation}\label{full_d_H0}
			\tilde{d}_{U}(\mu,  \beta, y)=   -\frac{2}{n} U^\prime   \mathcal{P}_Z \left( y-\Upsilon_{f,g}\mu-X\beta\right),
		\end{equation}
		where $\tilde{d}_U(0_{q\times 1}, \beta, y)= \tilde{d}$ defined in (\ref{d}).
		
		We partition $\hat{J}= (\hat{\Xi}, \hat{N})$, where $\hat{\Xi}$ and $\hat{N}$ are $m\times q$ and $m\times k$, respectively, with a similar partition for its expected value $J= (\Xi, N)$. Also, we define the $(q+k) \times (q+k)$ matrix $\hat{D}= \partial ^2\mathcal{Q}/\partial \gamma\partial \gamma^\prime$,  such that the first $q\times q$ block is given by
		\begin{equation}
			\hat{D}_{11}= \frac{2}{n} \Upsilon_{f,g}'  \mathcal{P}_Z \Upsilon_{f,g} = 2 \hat{\Xi}^\prime \hat{M}^{-1} \hat{\Xi}
		\end{equation}
		the block 1-2 (or the transposed of 2-1 block) is the $q\times k$ matrix
		\begin{equation}
			\hat{D}_{12} =\hat{D}_{21}^{\prime}= \frac{2}{n} \Upsilon_{f,g}'  \mathcal{P}_Z X   = 2 \hat{\Xi}^{\prime}\hat{M}^{-1} \hat{N}
		\end{equation}
		and the 2-2 block is the $k  \times k$ matrix
		\begin{equation}
			\hat{D}_{22} = \frac{2}{n} X'  \mathcal{P}_Z X = 2 \hat{N}^\prime \hat{M}^{-1} \hat{N}.
		\end{equation}
		Under Assumption \ref{ass:eigsandinsts}, $\Vert \hat{D} \Vert =O_p(1)$ and $\liminf_{n\rightarrow\infty}\underline{{eig}}\left(\hat{D} \right) >0$ with inverse defined and partitioned in the usual way. Also, $\hat{D}$ does not depend on any unknowns. In line we our previous notation, we also define the corresponding limit quantities as $D_{11}= 2 \Xi^{\prime} M^{-1} \Xi$, $D_{12}= D_{21}^\prime= \Xi^\prime M^{-1} N$ and $D_{22}= 2 N^\prime M^{-1} N$.
		
		From standard algebra, by the mean value theorem (MVT), given $\hat{d}$ in (\ref{dhat}),
		\begin{align}
			& \hat{d}_p= \left .\frac{\partial \mathcal{Q}}{\partial \mu^\prime}\right\vert _{(0_{1 \times q}, \hat{\beta}^\prime)^\prime} = \left. \frac{\partial \mathcal{Q}}{\partial \mu^\prime}\right\vert_{(0_{1\times q}, \beta_0^\prime)^\prime}+ \hat{D}_{12}(\hat{\beta}- \beta_0)\notag \\
			& 0= \left .\frac{\partial \mathcal{Q}}{\partial \beta^\prime}\right\vert _{(0_{1 \times q}, \hat{\beta}^\prime)^\prime} = \left. \frac{\partial \mathcal{Q}}{\partial \beta^\prime}\right\vert_{(0_{1\times q}, \beta_0^\prime)^\prime}+ \hat{D}_{22}(\hat{\beta}- \beta_0)
		\end{align}
		Thus,
		\begin{align}
			\hat{d}_p= \left(I_q ; \  -\hat{D}_{12}\hat{D}_{22}^{-1}\right)\left.\begin{pmatrix}\frac{\partial \mathcal{Q}}{\partial \mu^\prime} \\ \frac{\partial \mathcal{Q}}{\partial \beta^\prime} \end{pmatrix}\right\vert_{(0_{1\times q}, \beta_0^\prime)^\prime} =&  \left(I_q ; \  -\hat{D}_{12}\hat{D}_{22}^{-1}\right) \tilde{d}_{U}(0_{q\times 1}, \beta_0) \notag \\
			=& \left(I_q ; \  -\hat{D}_{12}\hat{D}_{22}^{-1}\right) \tilde{d}( \beta_0)
		\end{align}
		according to the definition in (\ref{full_d_H0}) and (\ref{d}), and with $I_q$ denoting the $q\times q$ identity matrix. Hence, given $\hat{H}$ in (\ref{Hhat}),
		\begin{equation}\label{equiv_th4}
			n\hat{d}_p^\prime \hat{H}^{11} \hat{d}_p = n \tilde{d}_U(0_{q\times 1}, \beta_0)^\prime \hat{\mathcal{V}}\tilde{d}_{U}(0_{q\times 1}, \ \beta_0),
		\end{equation}
		with 
		\begin{align}\label{italic_V}
			\hat{\mathcal{V}}=& \begin{pmatrix} I_{q} \\ - \hat{D}_{22}^{-1} \hat{D}_{21}\end{pmatrix} \hat{H}^{11} \left(I_q\ ; \ - \hat{D}_{12}\hat{D}_{22}^{-1} \right). 
		\end{align}
		
		\noindent Thus, 
		\begin{align}\label{stat1}
			n \hat{d}^\prime \hat{H}^{-1}\hat{d}= n\hat{d}_p^\prime \hat{H}^{11} \hat{d}_p = n \tilde{d}_{U}(0_{q\times 1}, \ \beta_0)^\prime \hat{\mathcal{V}} \tilde{d}_{U}(0_{q\times 1}, \ \beta_0) 
		\end{align}
		
		\noindent However,  under $\mathcal{H}_{1A}$, $\tilde{d}_{U}(0_{q\times 1}, \ \beta_0)$  is no longer evaluated at the true parameter value as $\mu_0 \neq 0$.  By MVT around $\mu_0$, we can write
		\begin{equation}\label{d_H1}
			\tilde{d}_U(0_{q\times 1}, \ \beta_0)= \tilde{d}_U(\mu_0,\ \beta_0)  - \frac{\partial \tilde{d}_U(\bar{\mu}, \ \beta_0)}{\partial \mu} \mu_0 \equiv \tilde{d}_U(\mu_0,\ \beta_0)   + \tau,
		\end{equation}
		with $\bar{\mu}$ being  intermediate point such that $\Vert \bar{\mu}-\mu_0 \Vert \leq \Vert \mu_0\Vert$ and $\tau$ being the $q+k \times 1$ vector defined as
		\begin{equation}
			\tau=	-\frac{\partial \tilde{d}_U(\bar{\mu}, \ \beta_0)}{\partial \mu} \mu_0= \frac{2}{n}U^\prime  \mathcal{P}_Z \Upsilon_{f,g}\mu_0 =  \hat{J}^\prime \hat{M}^{-1}   \hat{\Xi} \mu_0.
		\end{equation}
		Similarly to (\ref{theorem1_1}) and (\ref{d_remainder}), 
		\begin{align}
			\Vert \tilde{d}_U(\mu_0,\ \beta_0)\Vert \leq & K \Vert\hat{J} \Vert \Vert \hat{M}^{-1}\Vert \left \Vert \frac {1}{n} Z^\prime \epsilon \right\Vert +  K \Vert\hat{J} \Vert \Vert \hat{M}^{-1}\Vert \left \Vert \frac {1}{n} Z^\prime R \right\Vert \notag \\= &O_p\left(\max\left(\sqrt{\frac{p}{n}}, p^{-\nu} \right)\right) = O_p \left(\sqrt{\frac{p}{n}}\right)
		\end{align}
		for $\nu$ satisfying $\sqrt{n}/p^{\nu+1/2}=o(1)$, which holds under Assumption \ref{ass:eigsandinsts}, and $\Vert \tau \Vert=O_p(1)$ and non-zero, since $\mu_0\neq 0$. 

		We furthermore define the unconstrained version of $\hat{H}$ evaluated at generic parameters' value as
		\begin{equation}
			\tilde{H}_U(\mu,  \beta)= 4 \hat{J}^\prime \hat{M}^{-1} \tilde{\Omega}_U(\mu,  \beta) \hat{M}^{-1} \hat{J},  
		\end{equation}
		partitioned in the usual way, where $\tilde{\Omega}_U$ is defined according to (\ref{epsilonU_def}). We also define its limit quantity $H_U(\mu_0, \ \beta_0)= 4 J^\prime M^{-1} \Omega M^{-1} J$, where, as previously defined, $\Omega= n^{-1}\mathbb{E}(Z^\prime \Sigma Z)$ and $\Sigma$ is the $n\times n$ block-diagonal matrix with $n_g\times n_g$ diagonal block given by $\Sigma_g, g=1,\ldots,G$.
		Similar to earlier calculations in the proof of Theorem \ref{theorem:approx}, under Assumptions \ref{ass:eigsandinsts}-\ref{ass:power}, $\Vert \tilde{H}_U (\mu,  \beta)\Vert = O_p(1)$, uniformly in $(\mu,  \beta)$ and $\liminf_{n\rightarrow\infty}\underline{\textit{eig}}(\tilde{H}_U(\mu,  \beta))> c>0$, uniformly in $(\mu,  \beta)$ and almost surely. 
		
		Clearly, $\hat{H}= \tilde{H}_U(0,\hat{\beta})$. We can apply the MVT to $\hat{H}^{-1}$ around the true parameters' value and obtain
		\begin{align}\label{H_H1}
			\hat{H}^{-1}&= \tilde{H}_U^{-1}(\mu_0, \beta_0)+  \sum_{j=1}^p \tilde{H}_U^{-1}(\bar{\mu}, \bar{\beta}) \frac{\partial \tilde{H}_U}{\partial{\mu_j}}\vert_{(\bar{\mu}, \bar{\beta})} \tilde{H}_U^{-1}(\bar{\mu}, \bar{\beta}) \mu_{0j} \notag \\
			&- \sum_{t=1}^k \tilde{H}_U^{-1}(\bar{\mu}, \bar{\beta}) \frac{\partial \tilde{H}_U}{\partial{\beta_t}}\vert_{(\bar{\mu}, \bar{\beta})} \tilde{H}_U^{-1}(\bar{\mu}, \bar{\beta}) (\hat{\beta}_t - \beta_{0t}) 
			\equiv  \tilde{H}_U^{-1}(\mu_0, \beta_0) + T,
		\end{align}
		where $\bar{\mu}$ and $\bar{\beta}$ are intermediate points such that $\Vert \bar{\mu}- \mu_0\Vert \leq \Vert \mu_0\Vert$ and $\Vert \bar{\beta}- \beta_0 \Vert \leq \Vert \hat{\beta}- \beta_0 \Vert$.  Under $\mathcal{H}_{0A}$, $\Vert T\Vert = O_p(\sqrt{p/n}).$ Under $\mathcal{H}_{1A}$, $\mu_{0j} \neq 0$ for some $j=1,\ldots,q$ and, since  $\hat{\beta}_{t}$ for $t=1,\ldots,k$ are restricted estimates, $\hat{\beta}_t - \beta_{0t}=O_p(1)$ for some $t=1,\ldots,k$. Thus, under Assumptions \ref{ass:eigsandinsts}-\ref{ass:power}, $\Vert T \Vert =O_p(p)$ and $\liminf_{n\rightarrow\infty}\underline{\textit{eig}}(T)> c>0$. By partitioning $T$ in the usual way, we obtain $\hat{H}^{11} = \tilde{H}_U^{11}(\mu_0, \ \beta_0) + T_{11}$.  Also, let 
		\begin{equation}\label{V_def}
			\tilde{\mathcal{V}}(\mu_0, \ \beta_0)= \begin{pmatrix} I_q \\ - \hat{D}_{22}^{-1} \hat{D}_{21}\end{pmatrix} \tilde{H}_U^{11}(\mu_0, \ \beta_0) \left(I_q\ ; \ - \hat{D}_{12}\hat{D}_{22}^{-1} \right)
		\end{equation}
		and
		\begin{equation}\label{W_limit}
			\tilde{\mathcal{W}}= \begin{pmatrix} I_q \\ - \hat{D}_{22}^{-1} \hat{D}_{21}\end{pmatrix} T_{11} \left(I_q\ ; \ - \hat{D}_{12}\hat{D}_{22}^{-1} \right).
		\end{equation}
		
		
		\noindent From (\ref{d_H1}) and (\ref{H_H1}), (\ref{stat1}) becomes
		\begin{align}
			n\hat{d}_p^\prime \hat{H}^{11} \hat{d}_p =& n \tilde{d}_{U}(\mu_0, \ \beta_0)^\prime \tilde{\mathcal{V}}(\mu_0, \ \beta_0) \tilde{d}_{U}(\mu_0, \ \beta_0) + 2n \tau^\prime \tilde{\mathcal{V}}(\mu_0, \ \beta_0) \tilde{d}_{U}(\mu_0, \ \beta_0) \notag \\
			+& n \tau^\prime \tilde{\mathcal{V}}(\mu_0, \ \beta_0) \tau + n \tilde{d}_{U}(\mu_0, \ \beta_0)^\prime \tilde{\mathcal{W}} \tilde{d}_{U}(\mu_0, \ \beta_0) \notag\\
			&+ 2n \tau^\prime \tilde{\mathcal{W}} \tilde{d}_{U}(\mu_0, \ \beta_0) 
			+ n \tau^\prime \tilde{\mathcal{W}} \tau,
		\end{align}
		and thus
		\begin{align}\label{statistic_h1}
			\frac{n\hat{d}_p^\prime \hat{H}^{11} \hat{d}_p -q}{(2q)^{1/2}}&= \frac{n \tilde{d}_{U}(\mu_0, \ \beta_0)^\prime \tilde{\mathcal{V}}(\mu_0, \ \beta_0) \tilde{d}_{U}(\mu_0, \ \beta_0)-q}{(2q)^{1/2}}\notag\\
			& + \frac{\sqrt{2}n}{\sqrt{q}}\tau^\prime \tilde{\mathcal{V}}(\mu_0, \ \beta_0) \tilde{d}_{U}(\mu_0, \ \beta_0) \notag \\
			&+  \frac{n}{\sqrt{2q}}\tau^\prime \tilde{\mathcal{V}}(\mu_0, \ \beta_0) \tau + \frac{n}{\sqrt{2q}}\tilde{d}_{U}(\mu_0, \ \beta_0)^\prime \tilde{\mathcal{W}} \tilde{d}_{U}(\mu_0, \ \beta_0)\notag\\ 
			&+ \frac{\sqrt{2} n}{\sqrt{q}}\tau^\prime \tilde{\mathcal{W}} \tilde{d}_{U}(\mu_0, \ \beta_0) + \frac{ n}{\sqrt{2q}} \tau^\prime \tilde{\mathcal{W}} \tau 
		\end{align}
		By a similar argument adopted in the proof of Theorem \ref{theorem:dhatd}, we can show $\Vert \tilde{d}_U(\mu_0, \ \beta_0) - d_U \Vert = O_p(p^{3/2}/n)$, with $d_U= -2/n J^\prime M^{-1} Z^\prime \epsilon$ and $d_p= (I_q; -D_{12} D_{22}^{-1}) d_U $. Also, we can show
		\begin{equation}\label{H_U convergence}
			\Vert \tilde{H}_U(\mu_0, \ \beta_0) - H_U \Vert = O_p\left(\frac{p}{\sqrt{n}}\right),
		\end{equation}
		such that, under Assumptions \ref{ass:eigsandinsts}-\ref{ass:power},  $\Vert \tilde{H}^{11}_U(\mu_0, \ \beta_0) - H_U^{11} \Vert = O_p\left(p/\sqrt{n}\right)$.
		We show the claim in (\ref{H_U convergence}) by routine arguments as in (\ref{H_conv_first_1}) and (\ref{omega_terms}), after observing that 
		$\tilde{H}_U(\mu_0, \ \beta_0) = 4 \hat{J}^{\prime} \hat{M}^{-1} \tilde{\Phi}_R \hat{M}^{-1} \hat{J}$, with $\tilde{\Phi}_R = \tilde{\Phi} + \sum_{i=1}^{n} z_i z_i^\prime R_i^2 /n$, and 
		\begin{align}
			&\left\Vert \Delta^{\tilde{\Phi}_R}_\Phi \right\Vert \leq \left\Vert \Delta^{\tilde{\Phi}}_\Phi \right\Vert + \left\Vert \frac{\sum_{i=1}^n z_i z_i^\prime R_i^{2}}{n}\right\Vert
			=  O_p\left(\frac{p}{\sqrt{n}}\right) + \underset{1\leq i\leq n}{\sup} R_i^2 \ \Vert \hat{M} \Vert  \notag \\&=  O_p\left(\frac{p}{\sqrt{n}}\right) + O_p(p^{-2\nu})= O_p\left(\frac{p}{\sqrt{n}}\right),
		\end{align}
		where the last equality follows for $\nu$ satisfying $\sqrt{n}/p^{2\nu+1}=o(1)$, which holds under Assumption \ref{ass:eigsandinsts}.

		After showing, similarly to steps in the proof of Theorem \ref{theorem:approx}, that
		\begin{equation}
			\tilde{d}_{U}(\mu_0, \ \beta_0)^\prime \tilde{\mathcal{V}}(\mu_0, \ \beta_0) \tilde{d}_{U}(\mu_0, \ \beta_0) -  d_{p}^\prime {H}_U^{11}d_{p}=o_p\left(\frac{\sqrt{p}}{n}\right), 
		\end{equation}
		we conclude that the first term in (\ref{statistic_h1}) is $O_p(1)$, as shown in Theorem  \ref{theorem:nulldist}. 
		By standard norm inequalities, the second term in (\ref{statistic_h1}) is $O_p(\sqrt{n})$, the third is $O_p(n/\sqrt{p})$, the fourth is $O_p(p^{3/2})$,  the fifth is $O_p(p\sqrt{n})$ and  the sixth is $O_p(n\sqrt{p})$. The last term dominates the former five ones and thus, under $\mathcal{H}_{1A}$, for all $\eta>0$, $\mathbb{P}\left(|\mathcal{S}|^{-1} \leq \eta/n \sqrt{p} \right) \rightarrow 1$ as $n\rightarrow \infty $ and hence consistency of $\mathcal{S}$ follows.
	\end{proof}

\renewcommand{\thesection}{D}
\setcounter{lemma}{0} \renewcommand{\thelemma}{D\arabic{lemma}}
\renewcommand{\theequation}{D.\arabic{equation}}
\setcounter{equation}{0}

\section{Auxiliary lemmas}\label{appendix:lemmas}

\begin{lemma}\label{lemma:Mhat} 
	Let $p^2/n\rightarrow 0$ as $n\rightarrow\infty$ and suppose that Assumptions \ref{ass:regressors}-\ref{ass:Mhat} hold with $\nu>3/2$. Then, as $n\rightarrow \infty$, 
	\begin{equation}\label{MhatJhat}
		\left\Vert \hat{M} - M \right\Vert =O_p\left(\frac{p}{\sqrt{n}}\right), 		\left\Vert \hat{J} - J \right\Vert =O_p\left(\frac{p}{\sqrt{n}}\right).
	\end{equation}	
\end{lemma}

\begin{proof}[Proof of Lemma \ref{lemma:Mhat}:]
This is Lemma 1 in \cite{Guptaetal2024}.
\end{proof}
\begin{lemma}\label{lemma:Omegaeigs}
	Under Assumptions \ref{ass:errors} and \ref{ass:eigsandinsts}, 
	\[
	\limsup_{n\rightarrow\infty}\eigbig(\Phi)<\infty \text{ and } \liminf_{n\rightarrow\infty}\eigsmall(\Phi)>0.
	\]
\end{lemma}
\begin{proof}
Let $x$ be a non-stochastic $m\times 1 $ vector with $\Vert x \Vert=1$. Then
\[
x'\Phi x=\mathbb{E} \left(x'n^{-1}Z'\Sigma Zx\right)
\leq \mathbb{E}\left(x'n^{-1}Z'Z\eigbig(\Sigma)x\right)=(x'Mx)\eigbig(\Sigma)\leq \eigbig(M)\eigbig(\Sigma),
\]
uniformly over $x$ such that $\Vert x \Vert=1$. Then the claim for $\eigbig(\Phi)$ follows by (\ref{ass_6_Sigma}) and (\ref{ass_6_a}). The proof of the claim for $\eigsmall(\Phi)$ is similar.
\end{proof}



\renewcommand{\thesection}{E}
\setcounter{lemma}{0} \renewcommand{\thelemma}{E\arabic{lemma}}

\section{Original results}\label{app_C}

In this section, we reproduce and discuss the original empirical findings from Section \ref{illustrations}. These results serve as a reference for comparison with our nonparametric test results, allowing readers to assess how our procedure extends the original analyses.

\renewcommand{\thetable}{E\arabic{table}}
\renewcommand{\theequation}{E.\arabic{equation}}
\setcounter{table}{0}
\setcounter{equation}{0}

\subsection{Professional golf tournaments \citep{GuryanKroftNotowidigdo}}\label{app:gkn}

Table \ref{tab:GKNaugmented} reproduces the estimates in \citet{GuryanKroftNotowidigdo}, who study peer effects in professional golf tournaments.
Column~(i) reports the authors' baseline regression from
\begin{equation}
  y_{i,tr} = \alpha + \beta Ability_i + \gamma w_{i,tr}'Ability + \delta_{tc} + \varepsilon_{i,tr},
  \label{eq:gkn-baseline}
\end{equation}
where $y_{i,tr}$ is player~$i$'s score in round~$r$ of tournament~$t$; $Ability_i$
is own ability, measured by the average corrected handicap score over the last 2-3 years; and $w_{i,tr}'Ability$ is peer ability, defined below.\footnote{The authors therefore construct a simplified version of the United States Golf Association (USGA) handicap correction as a measure of ability. For more details on how this variable was originally constructed, see 
\citet{GuryanKroftNotowidigdo}.}


Let $W$ denote a peer-weight matrix on the stacked player--group--round--tournament observations, with generic element $w_{i,tr,j}$ giving the weight assigned to player~$j$ in player~$i$'s group at round~$r$ and tournament~$t$. The vector $w_{i,tr}$ is the corresponding row of $W$, and $w_{i,tr}'Ability \equiv \sum_{j \neq i} w_{i,tr,j}\,Ability_{t,j}$ is the weighted mean peer ability faced by player~$i$. The coefficients $\beta$ and $\gamma$ capture the effects of own and peer ability on scores; $\delta_{tc}$ denotes tournament-by-category fixed effects; and $\varepsilon_{i,tr}$ is an idiosyncratic error term.

Column (ii) considers alternative measures of playing partners' ability that may influence performance through different channels. Specifically, it replaces partners' handicap with measures such as \textit{driving distance}, number of \textit{putts}, and \textit{greens} hit in regulation, which help distinguish potential `learning' effects from pure `motivation' effects.\footnote{The intuition is that players may learn about wind or course conditions from observing another player's putting, but cannot directly learn to drive longer; the driving-distance coefficient thus captures the motivation component net of learning.} 
In Column~(iii), the authors also include interaction terms allowing peer effects to vary with player's own ability (based on the corrected handicap) and with experience (measured as the number of years the player has competed
professionally on the PGA Tour).

\begin{table}[p]

\centering
\footnotesize
\caption{Original results by \cite{GuryanKroftNotowidigdo}}
\label{tab:GKNaugmented}
\begin{tabular}{lccc}
\hline 
Dependent var. & (i) &  (ii) & (iii) \\
\hline
$Ability_{i}$ & 0.672\({}^{***}\)   &  & 0.656\({}^{***}\)  \\
& (0.039)  &  & (0.039) \\[0.5em]
$w_{i,tr}' Ability$ & -0.035  & & -0.036  \\
& (0.040)  & & (0.040) \\[0.5em]
$DrivDist_{i}$ &  & -0.009 & \\
&  & (0.004) & \\[0.5em]
$w_{i,tr}'  DrivDist$ & & 0.003 & \\
& & (0.004) &  \\[0.5em]
$Putts_{i}$ &  & 0.130\({}^{***}\) & \\
&  & (0.030) & \\[0.5em]
$w_{i,tr}'  Putts$ &  & -0.045 & \\
&  & (0.039) & \\[0.5em]
$Greens_{i}$ &  & -0.682\({}^{***}\)  & \\
&  & (0.050) & \\[0.5em]
$w_{i,tr}' Greens$ &  & -0.023 & \\
& &  (0.060) & \\[0.5em]
$Ability_{i} \times w_{i,tr}' Ability$  & & & 0.081 \\
& &   & (0.033) \\[0.5em]
$Exp_{it}$ & & & 0.019\({}^{***}\) \\
& &   & (0.004) \\ [0.5em]
$Exp_{it} \times w_{i,tr}' Ability $ & & & 0.015\({}^{**}\)  \\
& &   & (0.005) \\ 
\hline
$n$ & 17,492& 17,182 & 17,492 \\
\hline 
\end{tabular}
\begin{minipage}[t]{0.6\textwidth}
 \caption*{\footnotesize \textit{Notes}: Regression results in  Columns (i) and (ii) replicate Columns 1 and 5 from Table 5 in \cite{GuryanKroftNotowidigdo}. Column (iii) replicates Column 4 from Table 8 in \cite{GuryanKroftNotowidigdo}. The dependent variable is the golf score of player $i$ in a given round. $Ability_{i}$ denotes player $i$’s ability, measured by the average handicap score of the last 2-3 years. $w_{i,tr}' Ability = \sum_{j \neq i} w_{i,tr,j} Ability_j$ denotes the weighted average ability of player $i$’s peers $j$. Column (ii) includes additional ability measures such as driving distance ($DrivDist_i$), number of putts ($Putts_i$), and greens in regulation ($Greens_i$), together with the corresponding peer measures $w_{i,tr}' DrivDist$, $w_{i,tr}' Putts$, and $w_{i,tr}' Greens$.
 Column~(iii) further includes interactions between own ability and peer ability, $Ability_i \times w_{i,tr}' Ability$, and between experience and peer ability, $Exp_{it} \times w_{i,tr}' Ability$, where $Exp_{it}$ denotes player experience. All specifications include tournament-by-category fixed effects. Standard errors are clustered at the playing-group level. Observations are weighted by the inverse of the sample variance of the ability measure.
$^{*}p<0.10$, $^{**}p<0.05$, $^{***}p<0.01$.}
\end{minipage}
\end{table}
\floatpagestyle{empty}

Overall, these results based on linear specifications find 
limited evidence of peer effects in individual performance.
None of the peer coefficients in specifications (i) and (ii)
 are statistically significant. In specification (iii), which also includes interaction terms allowing peer effects to vary with a player's own ability and with experience, some evidence of heterogeneity emerges.   Interacting average peer ability  with own experience yields a small but statistically significant effect, 
suggesting that more experienced players respond more to their co-workers' ability.


\

\subsection{Interracial contact, stereotypes, and academic performance \citep{CornoLaFerraraBurns}}\label{app:clb}

Table \ref{tab:CLBregress} reproduces the academic performance estimates of Table 4 in \citet{CornoLaFerraraBurns}, which serve as the basis for the tests in Table \ref{tab:CLBtest}. For each academic outcome, the authors estimate

\begin{equation}\label{eq:clb}
y_{ik} = \alpha +  \lambda w_i'Race  + x_{i}'\beta + \mu w_i'x + \delta_{k} + \varepsilon_{ik},
\end{equation}
where $y_{ik}$ is the outcome of student $i$ in residence $k$, measured at the end of the first academic year;  $Race$ is a vector of race dummies of all students in the sample; $x_{i}$ is a set of individual baseline controls (gender, UCT admission score, household wealth, monthly consumption, foreign status, and private high school attendance); $x$ is a matrix stacking the vectors $x_i$ for all respective $i$; $\delta_k$ denotes residence fixed effects;
and $\varepsilon_{ik}$ is an idiosyncratic error term.  Let $W$ denote the room adjacency matrix, with generic element $w_{ij}=1$ if students $i$ and $j$ share a  room at the beginning of the academic year and $0$ otherwise, and let $w_i$ denote the $i$th row of $W$ (which selects student $i$'s roommate).  The specification is estimated separately on the White subsample, the Black subsample, and the full sample.\footnote{
In the original paper, the authors' treatment indicator is
$\textit{MixRoom}_i=|Race_i-w_i'Race|$, which equals one when student $i$
and their roommate are of different races. Within each race subsample,
this indicator is an affine function of the peer-race attribute $w_i'Race$ that we use in our notation. 
In the full sample, where own race varies, $MixRoom_i$ differs from
$w_i'Race$. Thus we obtain it directly as
$w_i^{\mathrm{race}}\mathbf{1}_n=MixRoom_i$ using  $W^{\mathrm{race}}$, which is the restriction of $W$ to different-race
roommate pairs.\label{footnote:w_corno}}
\begin{table}[p]
\footnotesize
\centering
\caption{Original results by \cite{CornoLaFerraraBurns}}\label{tab:CLBregress}
\resizebox{0.8\textwidth}{!}{%
\begin{tabular}{lcccc}
\toprule
Dependent var. & GPA & \begin{tabular}[c]{@{}c@{}}Number of\\ exams passed\end{tabular} & \begin{tabular}[c]{@{}c@{}}Eligible to\\ continue\end{tabular} & \begin{tabular}[c]{@{}c@{}}Academic\\ performance index\end{tabular} \\
 & (i) & (ii) & (iii) & (iv) \\
\midrule
\multicolumn{5}{l}{\textit{Panel A: Whites}}\\
$w_i'Race$ & $-$0.028 & $-$0.168 & 0.050 & 0.010 \\
 & (0.243) & (0.523) & (0.066) & (0.259) \\[0.5em]
$n$ & 117 & 117 & 117 & 117 \\
\midrule
\multicolumn{5}{l}{\textit{Panel B: Blacks}}\\
$w_i'Race$ & 0.257\({}^{**}\) & 0.645\({}^{***}\) & 0.152\({}^{***}\) & 0.443\({}^{***}\) \\
 & (0.125) & (0.245) & (0.040) & (0.141) \\[0.5em]
$n$ & 332 & 332 & 332 & 332 \\
\midrule
\multicolumn{5}{l}{\textit{Panel C: Full sample}}\\
$w_i^{\mathrm{race}}\mathbf{1}$ & 0.147 & 0.447\({}^{**}\) & 0.105\({}^{***}\) & 0.289\({}^{**}\) \\
 & (0.102) & (0.204) & (0.031) & (0.113) \\[0.5em]
$n$ & 499 & 499 & 498 & 498 \\
\bottomrule
\end{tabular}
}
\begin{minipage}[t]{0.8\textwidth}
\caption*{\footnotesize Note: Results originally presented in Table 4 of \cite{CornoLaFerraraBurns}. The dependent variable in column (i) is GPA (standardized over the full sample); in column (ii) it is the number of exams passed during the first year; in column (iii) it is a dummy for being in good standing and eligible to continue the following year; in column (iv) it is an index constructed as the first principal component of the previous three variables. Controls are measured at baseline and follow equation \eqref{eq:clb}. Individual and roommate controls include gender, UCT admission score, household wealth, monthly consumption, foreign status, and private high school attendance. Standard errors are clustered at the room level. $^{*}p<0.10$, $^{**}p<0.05$, $^{***}p<0.01$.}
\end{minipage}
\end{table}
The pattern in Table \ref{tab:CLBregress} is heterogeneous across groups. For White students (Panel A), the estimated effects of being in a mixed room are close to zero and statistically insignificant for all four outcomes. For Black students (Panel B), the effects are positive, sizable, and significant throughout.
In the full sample (Panel C), the estimates are positive and significant for exams passed, continuation, and the performance index, but insignificant for GPA.

\subsection{Ability mix and student performance \citep{wuzhangwang2023} }\label{app:wzw}

Table \ref{tab:peer_effects_mixed_seating} reproduces the deskmate-level peer-effect estimates of Table~5 in \citet{wuzhangwang2023}, which serve as the basis for the tests in Table \ref{tab:WZWtest}. In the mixed-seating (MS) and mixed-seating-with-reward (MSR) classes, students are randomly paired as deskmates within height groups. Students completed baseline (pre-intervention) and endline (post-intervention) surveys of academic performance and personality traits. For each endline outcome, the authors estimate

\begin{equation}\label{eq:wzw}
y_{i}^{\mathrm{end}} = \alpha + \beta y_{i}^{\mathrm{base}} + \gamma\, w_i'y^{\mathrm{base}} + x_i'\delta + h_{m(i)} + \varepsilon_i,
\end{equation}
where $y_{i}^{\mathrm{end}}$ is student~$i$'s endline outcome---either the average academic z-score (column~i) or one of the `big five' personality traits (columns~ii--vi)---and $y_{i}^{\mathrm{base}}$ is the corresponding baseline measure. Let $W$ denote the deskmate adjacency matrix on the estimation sample, with generic element $w_{ij}=1$ if $j$ is student~$i$'s assigned deskmate and $0$ otherwise. The vector $w_i$ is the corresponding row of $W$, so that $w_i'y^{\mathrm{base}}$ is the deskmate's baseline value of the same outcome. The vector $x_i$ collects individual controls (gender, age, height, health status, hukou registration, minority status, parental education, and household assets); $h_{m(i)}$ denotes class-by-height-group fixed effects; and $\varepsilon_i$ is an idiosyncratic error term.  The specification is estimated separately for lower- and upper-track students in the MS and MSR classes (Panels A--D).

\begin{table}[p]\centering
\caption{Original results by \cite{wuzhangwang2023}}
\label{tab:peer_effects_mixed_seating}
\footnotesize
\setlength{\tabcolsep}{4pt}
\begin{tabular}{
p{3.2cm}
c
c
c
c
c
c
}
\toprule
& \begin{tabular}[c]{@{}c@{}}Average\\ z-scores\\ (i)\end{tabular}
& \begin{tabular}[c]{@{}c@{}}Extraver-\\ sion \\ (ii)\end{tabular}
& \begin{tabular}[c]{@{}c@{}}Agreeable-\\ ness \\ (iii)\end{tabular}
& \begin{tabular}[c]{@{}c@{}}Openness \\ (iv)\end{tabular}
& \begin{tabular}[c]{@{}c@{}}Neuroticism \\ (v)\end{tabular}
& \begin{tabular}[c]{@{}c@{}}Conscien-\\ tiousness \\ (vi)\end{tabular} \\

\midrule

\multicolumn{7}{l}{\textit{Panel A. Lower-track students in MS classes}} \\
$w_i'y^{\mathrm{base}}$
& $-0.029$ & $0.054$ & $-0.059$ & $-0.048$ & $0.026$ & $0.044$ \\
& $(0.106)$ & $(0.044)$ & $(0.038)$ & $(0.057)$ & $(0.051)$ & $(0.049)$ \\
$n$
& 317 & 317 & 317 & 317 & 317 & 317 \\

\midrule

\addlinespace
\multicolumn{7}{l}{\textit{Panel B. Upper-track students in MS classes}} \\
$w_i'y^{\mathrm{base}}$
& $-0.032$ & $0.054$ & $0.084$ & $0.055$ & $0.095^{*}$ & $0.074$ \\
& $(0.034)$ & $(0.044)$ & $(0.054)$ & $(0.078)$ & $(0.052)$ & $(0.092)$ \\
$n$
& 317 & 317 & 317 & 317 & 317 & 317 \\ 

\midrule

\addlinespace
\multicolumn{7}{l}{\textit{Panel C. Lower-track students in MSR classes}} \\
$w_i'y^{\mathrm{base}}$
& $0.049$ & $0.225^{*}$ & $0.110^{**}$ & $0.108^{*}$ & $-0.075$ & $-0.022$ \\
& $(0.050)$ & $(0.124)$ & $(0.053)$ & $(0.061)$ & $(0.101)$ & $(0.086)$ \\
$n$ & 297 & 297 & 297 & 297 & 297 & 297 \\

\midrule
\addlinespace
\multicolumn{7}{l}{\textit{Panel D. Upper-track students in MSR classes}} \\
$w_i'y^{\mathrm{base}}$
& $0.032$ & $0.152$ & $0.168^{*}$ & $-0.089$ & $0.058$ & $-0.021$ \\
& $(0.038)$ & $(0.112)$ & $(0.099)$ & $(0.067)$ & $(0.059)$ & $(0.077)$ \\
$n$ & 297 & 297 & 297 & 297 & 297 & 297 \\

\bottomrule
\end{tabular}

\begin{minipage}[t]{0.9\textwidth}
\caption*{\footnotesize Note: Results originally presented in Table~5 of \cite{wuzhangwang2023}. The dependent variables are the student’s endline average z-scores (column i) or `big five' personality traits (columns ii–vi). Regressions are run for upper- and lower-track students in the MS and MSR classes separately. Controls include own baseline performance, gender, age, height, health status, hukou registration status, minority status, father’s education, mother’s education, and whether the student’s
household has a computer or a car. Standard errors are clustered at the class level. $^{*}p<0.10$, $^{**}p<0.05$, $^{***}p<0.01$.}
\end{minipage}
\end{table}

Table \ref{tab:peer_effects_mixed_seating} shows that, in MS classes (Panels A--B), deskmate baseline performance is essentially unrelated to endline outcomes: all coefficients are small and insignificant at the 5\% level (with only neuroticism among upper-track students significant at 10\%). Under MSR (Panels C--D), lower-track students show positive associations between deskmate baseline traits and own Extraversion, Agreeableness, and Openness (the Agreeableness coefficient is significant at 5\%; Extraversion and Openness at 10\%), while the academic z-score coefficient remains positive but insignificant at conventional levels. For upper-track MSR students, only Agreeableness responds (at 10\%). Overall, the linear specifications detect little peer dependence of deskmate baseline achievement on academic scores.

\subsection{Peer effects in academic research \citep{bosquetcombesetal2022}}\label{app:bch}

Table \ref{tab:peer_effects_ols} reproduces the OLS peer-effect estimates of Table~3 in \citet{bosquetcombesetal2022}, which serve as the basis for the tests in Table \ref{tab:BCHtest}. We first illustrate the intermediate specification (iii) which relates individual research output to the number and  quality of  peers as

\begin{equation}\label{eq:bch}
y_{ift} = \alpha + \lambda w_{i,ut}'Peer + \gamma w_{i,uft}'Output  +  x_{i,uft}'\beta + \varepsilon_{ift},
\end{equation}
where $y_{ift}$ is researcher~$i$'s output in JEL field~$f$ at date~$t$.\footnote{this is measured as the three-year moving average of the number of articles divided by the number of authors, using publications dated $t{+}1$, $t{+}2$ and $t{+}3$ to account for publication lag.} 
Let $w_{i,ut}$ denote a row of the matrix collecting weights on other members of university $u$ in the same year, and similarly let $w_{i,uft}$ collect weights on other members of university $u$ in the same year and JEL code.
Then $w_{i,ut}'Peer$ is the (normalized) number of peers per year and $w_{i,uft}'Output$ is the average productivity of those peers per year and JEL code, where each peer's productivity is defined as career-average publications.\footnote{Note that own and peer measures share the year index~$t$ but are not contemporaneous publication windows, because own measures are three-year moving average while peer measures are career-average publications statistics.} 
The vector $x_{i,uft}$ collects own age together with the fixed effects used in the original specification: individual fixed effects, university fixed effects, and year--JEL fixed effects; $\varepsilon_{ift}$ is an idiosyncratic error term.

In the specification of 
Column (i) the measurement of individual research output is taken at the university-year level rather than at the university-year-field level: that is, the dependent variable $y_{it}$ is individual output aggregated over all research fields, and the same for productivity of peers $w_{i,ut}'Output$. 

In the specification of 
Column (ii) the aggregation takes place at the university-year-field level for the dependent  variable $y_{ift}$, but the productivity of peers $w_{i,ut}'Output$ is still aggregated at university-year level only.

Columns~(iv)--(vi) extend column (iii) by further interacting $w_{i,uft}'Output$ with an indicator for female receivers and/or with (centred) age, allowing spillovers to vary by receiver characteristics.\footnote{Our tested university-level sample (Table \ref{tab:BCHtest}) contains 42,861 author--year observations, compared with the 42,521 reported by \citet{bosquetcombesetal2022} (Table \ref{tab:peer_effects_ols}). The difference arises because \citet{bosquetcombesetal2022} compute peer averages using the full university--year roster before dropping observations with missing age. To satisfy the SAR requirement of a square weights matrix, we retain all individuals used to construct peer averages while estimating the model on the restricted sample. This preserves the original peer averages and yields only negligible differences from the published estimates.
}

\begin{table}[!htbp]\centering
\caption{Original results by \cite{bosquetcombesetal2022}}
\label{tab:peer_effects_ols}
\footnotesize
\setlength{\tabcolsep}{7pt}
\begin{tabular}{lcccccc}
\toprule
& (i) & (ii) & (iii) & (iv)  &(v) & (vi)  \\
& Tot. & JEL & JEL & JEL & JEL & JEL \\
\midrule

$w_{i,ut}'Peer$
& $0.092^{*}$ & $0.000$ & $0.000$ & $0.000$ & $0.001$ & $0.001^{**}$ \\
& $(0.047)$ & $(0.000)$ & $(0.001)$ & $(0.001)$ & $(0.000)$ & $(0.000)$ \\

$w_{i,ut}'Output$
& $0.055$ & $-0.001$ & &  &  &  \\
& $(0.265)$ & $(0.002)$ &  &  &  &  \\

$w_{i,uft}'Output$ 
& & & $0.623^{***}$ & $0.693^{***}$ & $0.607^{***}$ & $0.697^{***}$ \\
& & & $(0.016)$ & $(0.017)$ & $(0.014)$ & $(0.016)$ \\

$Woman_i \times w_{i,uft}'Output$
& & & & $-0.316^{***}$ & & $-0.418^{***}$ \\
& & & & $(0.022)$ & & $(0.022)$ \\

$Age_i \times w_{i,uft}'Output$
& & & & & $-0.022^{***}$ & $-0.025^{***}$ \\
& & & & & $(0.001)$ & $(0.002)$ \\

\midrule
$n$

& $42{,}521$ & $771{,}498$ & $771{,}498$ & $771{,}498$ & $771{,}498$ & $771{,}498$ \\
\bottomrule
\end{tabular}
\begin{minipage}[t]{0.95\textwidth}
\caption*{\footnotesize \textit{Notes:} Results originally presented in Table~3 of \cite{bosquetcombesetal2022}. The dependent variable is the yearly individual research output (three-year moving average of articles divided by number of authors), aggregated over fields in column~(i) and measured at the JEL-code level from column~(ii) on. The term $w_{i,ut}'Peer$ s is the (normalized) number of peers at the university-year level. 
The terms $w_{i,ut}'Output$ and $w_{i,uft}'Output$ 
are the average peer output at the university-year level and university--year--JEL level respectively. All regressions include age (centered at the sample mean), individual, and university fixed effects; column~(i) includes year fixed effects and columns~(ii)--(vi) year--JEL fixed effects. Standard errors are clustered at the university--year level.
$^{*}p<0.10$, $^{**}p<0.05$, $^{***}p<0.01$.}
\end{minipage}
\end{table}

Table \ref{tab:peer_effects_ols} shows no robust peer effects when peers are defined as the entire university: in columns~(i)--(ii), both the peer-count and average peer-output coefficients are small and insignificant at the 5\% level. Once peer output is measured within the same JEL field (columns~iii--vi), the linear specifications detect significant positive spillovers: one additional publication by field peers raises own field output by about 0.6--0.7 publications. The interactions indicate that these gains are smaller for women and for older researchers. The peer-count term remains close to zero throughout, so the authors conclude that same-field colleagues' productivity---rather than department size---is the relevant channel.

\end{document}